\documentclass[runningheads]{llncs}
\usepackage{graphicx}
\usepackage{comment}
\usepackage{tikz}
\usepackage{float}
\usepackage{hyperref}
\usepackage{amsmath,mathabx}
\usepackage[shortlabels]{enumitem}
\usepackage{algorithm}
\usepackage[noend]{algpseudocode}
\usepackage{wrapfig}
\input{insbox}
\usepackage{mwe}
\usepackage{graphicx}
\usepackage{subcaption}
\usepackage{pdfpages}

\def\BState{\State\hskip-\ALG@thistlm}

\hypersetup{
    colorlinks=true,
    linkcolor=blue,
    filecolor=magenta,      
    urlcolor=cyan,
    pdftitle={Overleaf Example},
    pdfpagemode=FullScreen,
    }
    
\makeatletter
\newenvironment{breakablealgorithm}
  {% \begin{breakablealgorithm}
   \begin{center}
     \refstepcounter{algorithm}% New algorithm
     \hrule height.8pt depth0pt \kern2pt% \@fs@pre for \@fs@ruled
     \renewcommand{\caption}[2][\relax]{% Make a new \caption
       {\raggedright\textbf{\fname@algorithm~\thealgorithm} ##2\par}%
       \ifx\relax##1\relax % #1 is \relax
         \addcontentsline{loa}{algorithm}{\protect\numberline{\thealgorithm}##2}%
       \else % #1 is not \relax
         \addcontentsline{loa}{algorithm}{\protect\numberline{\thealgorithm}##1}%
       \fi
       \kern2pt\hrule\kern2pt
     }
  }{% \end{breakablealgorithm}
     \kern2pt\hrule\relax% \@fs@post for \@fs@ruled
   \end{center}
  }
\makeatother

\newcommand{\raisedtarget}[1]{%
  \raisebox{\fontcharht\font`P}[0pt][0pt]{\hypertarget{#1}{}}%
}

\newtheorem{remark2}{Remark}
\newtheorem{statement}{Statement}
\newtheorem{scenario}{Scenario}
\newtheorem{invariant}{Invariant}

\newcommand{\bandwidth}{\mathrm{Bandwidth}}

\newcommand{\ap}[1]{}
\newcommand{\va}[1]{}
\newcommand{\is}[1]{}
\newcommand{\stefan}[1]{}

\newcommand{\None}{\mathrm{None}}

\begin{document}
\title{Self-Adjusting Linear Networks\\ with Ladder Demand Graph}
% \thanks{Supported by organization x.}}
%
\titlerunning{Self-Adjusting Linear Networks with Ladder Demand Graph}
\author{Vitaly Aksenov\inst{1} \and
Anton Paramonov\inst{1} \and
Iosif Salem\inst{2} \and
Stefan Schmid\inst{2}}
\authorrunning{V. Aksenov et al.}
\institute{ITMO University, St. Petersburg, Russia\\
%\email{aksenov@itmo.ru} \and
% St Petersburg University, St. Petersburg, Russia
% \email{anton.paramonov2000@gmail.com}\\
% % \url{http://www.springer.com/gp/computer-science/lncs} 
\and
TU Berlin, Berlin, Germany\\
% \email{iosif.salem@inet.tu-berlin.de, 
% stefan.schmid@tu-berlin.de}
}
%
%\author{}
%\institute{}

\maketitle              % typeset the header of the contribution
\begin{abstract}
Self-adjusting networks (SANs) have the ability to adapt to communication demand by dynamically adjusting the workload (or demand) embedding, i.e., the mapping of communication requests into the network topology. SANs can thus reduce routing costs for frequently communicating node pairs by paying a cost for adjusting the embedding. This is particularly beneficial when the demand has structure, which the network can adapt to. Demand can be represented in the form of a demand graph, which is defined by the set of network nodes (vertices) and the set of pairwise communication requests (edges). Thus, adapting to the demand can be interpreted by embedding the demand graph to the network topology. This can be challenging both when the demand graph is known in advance (offline) and when it revealed edge-by-edge (online). The difficulty also depends on whether we aim at constructing a static topology or a dynamic (self-adjusting) one that improves the embedding as more parts of the demand graph are revealed. Yet very little is known about these self-adjusting embeddings.

In this paper, the network topology is restricted to a line and the demand graph to a ladder graph, i.e., a $2\times n$ grid, including all possible subgraphs of the ladder. We present an online self-adjusting network that matches the known lower bound asymptotically and is 12-competitive in terms of request cost. As a warm up result, we present an asymptotically optimal algorithm for the cycle demand graph. We also present an oracle-based algorithm for an arbitrary demand graph that has a constant overhead.

\keywords{Ladder graph \and Self-adjusting networks \and Traffic patterns \and online algorithms.}
\end{abstract}

\vspace{-1cm}
\section{Introduction}
\vspace{-0.2cm}

Traditional networks are static and demand-oblivious, i.e., designed without considering the communication demand. 
% They are usually optimized for worst-case metrics like diameter, bisection bandwidth, or resilience to failures. 
While this might be beneficial for all-to-all traffic, it doesn't take into account temporal or spatial locality features in demand. That is, sets of nodes that temporarily cover the majority of communication requests may be placed diameter-away from each other in the network topology.
This is a relevant concern as studies on datacenter network traces have shown that communication demand is indeed bursty and skewed \cite{avin2020complexity}.

Self-adjusting networks (SANs) are optimized towards the traffic they serve. SANs can be static or dynamic, depending on whether it is possible to reconfigure the embedding (mapping of communication requests to the network topology) in between requests, and offline or online, depending on whether the sequence of communication requests is known in advance or revealed piece-wise.
In the online case, we assume that the embedding can be adjusted in between requests at a cost linear to the added and deleted edges, thus, bringing closer frequently communicating nodes.
Online algorithms for SANs aim to reduce the sum of routing and reconfiguration (re-embedding) costs for any communication sequence.

We can express traffic in the form of a demand graph that is defined by the set of nodes in the network and the set of pairwise communication requests (edge set) among them. 
% Online SANs have been mostly designed for arbitrary demand patterns, i.e., arbitrary demand graphs.
% That is, in most proposals the SAN aims to react to arbitrary locality patterns \cite{DBLP:journals/ton/SchmidASBHL16,DBLP:conf/infocom/AvinS020} or assume sparse communication demand \cite{DBLP:conf/apocs/AvinS21}. 
Knowing the structure of the demand graph could allow us to further optimize online SANs, even though the demand is still revealed online.
That is, by re-embedding the demand graph to the network we optimize the use of network resources according to recent patterns in demand.

To the best of our knowledge, the only work on demand graph re-embeddings to date is \cite{avin2019self}, where the network topology is a line and the demand graph is also a line.
The authors presented an algorithm that serves $m = \Omega(n^2)$ requests at cost $O(n^2 \log n + m)$ and showed that this complexity is the lower bound. The problem is inspired by the Itinerant List Update Problem~\cite{olver2018itinerant} (ILU). To be more precise, the problem in~\cite{avin2019self} appears to be the restricted version of the online Dynamic Minimum Linear Arrangement problem, which is another reformulation of ILU.

\textbf{Contributions.} In this work, we take the next step towards optimizing online SANs for more general demand graphs. We restrict the network topology to a line, but assume that the demand graph is a ladder, i.e., a $2\times n$ grid.
We assume that before performing a request, we can re-adjust the line graph by performing several swaps of two neighbouring nodes, paying one for each swap. 
We present a 12-competitive online algorithm that embeds a ladder demand graph to the line topology, thus, asymptotically matching the lower bound in \cite{avin2019self}. This algorithm can be applied to any demand graph that is a subgraph of the ladder graph and that when all edges of the demand graph are revealed the topology is optimal and no more adjustments occur. We also optimally solve the case of cycle demand graphs, which is a simple generalization of the line demand graph, but is not a subcase of the ladder due to odd cycles. Finally, we provide a generic algorithm for arbitrary demand graphs, given an oracle that computes an embedding with the cost of requests bounded by the bandwidth.% of the demand graph. 

A solution for the ladder is the first step towards the $n \times m$ grid demand graph. 
% One of the main arising problems is that even an embedding of a tree onto an infinite grid is NP-hard~\cite{heckmann1991optimal}, while it is polynomial in our case. 
% leading to a reasonably simple algorithm. 
Moreover, a ladder (and a cycle) has a constant \emph{bandwidth}, i.e., a minimum value over all embeddings onto a target line graph of a maximal path between the ends of an edge (request). It can be shown that given a demand graph $G$ the best possible complexity per request is the bandwidth.

\textbf{Related work.}
Avin et al. \cite{avin2019self}, consider a fixed line (host) network and a line demand graph. Their online algorithm re-embeds the demand graph to the host line graph with minimum number of swaps on the embedding.
Both \cite{avin2022deterministic,avin2022push} present constant-competitive online algorithms for a fixed and complete binary tree, where nodes can swap and the demand is originating only from the source. However, these two works do not consider a specific demand graph.
Moreover, \cite{avin2022demand} studied optimal but static and bounded-degree network topologies, when the demand is known.
Self-adjusting networks have been formally organized and surveyed in \cite{avin2019toward}.
Other existing online SAN algorithms consider different models. The most distinct difference is our focus on online re-embedding while keeping a fixed host graph (i.e., a line) compared to other works that focus on changing the network topology. The latter is what, for example, SplayNet \cite{DBLP:journals/ton/SchmidASBHL16} is proposing, where tree rotations change the form of the binary search tree network, without optimizing for a specific family of demand patterns.

Online demand graph re-embedding also relates to dynamically re-allocating network resources to follow traffic patterns. In \cite{avin2016online}, the authors consider a fixed set of clusters of bounded size, which contain all nodes and migrate nodes online according to the communication demand. But more broadly, \cite{batista2007self} assumes a fixed grid network and migrates tasks according to their communication patterns.

% Self-adjusting networks (SANs) can be seen from the point of view of physical topology reconfiguration (e.g., via optical switches \cite{DBLP:conf/sigcomm/GhobadiMPDKRBRG16}) or that of virtualization (e.g., process migration).
% Self-adjusting networks (SANs) \stefan{Improve discussion} have been made possible due to novel network technologies \cite{DBLP:conf/sigcomm/GhobadiMPDKRBRG16}. 
% Their formal representations are introduced in  \cite{DBLP:journals/ccr/AvinS18}, \cite{DBLP:journals/ton/SchmidASBHL16}. 
% Existing works have studied linear networks \cite{avin2019self}, bounded degree networks \cite{DBLP:conf/apocs/AvinS21}, and more complex topologies that are not related to this work.
% Moreover, \cite{avin2022demand} studied optimal network topologies, when the demand is known.
% skip list networks \cite{DBLP:conf/infocom/AvinS020} and skip graphs \cite{DBLP:conf/icdcs/HuqG17}. 
% Moreover, recent work has also focused in more \stefan{these are very different models. clarify please} advanced cost models \cite{DBLP:journals/corr/abs-2006-11148}, \cite{DBLP:conf/spaa/FederRSSAS021}.
% The most relevant work to ours \cite{avin2019self} studied linear networks with linear demand graphs.
Also, relevant problems, from a migration point of view, are the classic list update problem (LU) \cite{sleator1985amortized}, the related Itinerant List Update (ILU) problem \cite{olver2018itinerant}, and the Minimum Linear Arrangement (MLA) problem \cite{hansen1989approximation}.
In contrast to those problems, we study an online problem where requests occur between nodes.

\textbf{Roadmap. }
Section~\ref{sec:model} describes the model and background. Section \ref{sec:results} contains the summary of our three contributions (ladder, cycle, general demand graph) and their high-level proofs. 
% This list contains the result for the cycle, the ladder, and the general demand graph. 
% In Section~\ref{sec:notions} we introduce all the necessary notions for our algorithm on the ladder.  
Section~\ref{sec:algorithm} presents the algorithm and analysis for ladder demand graphs. 
% Specifically, we calculate the complexity of the presented algorithm and show that it achieves the desired cost.
% Section~\ref{sec:conclusion} concludes the paper. 
Some technical details are deferred to the appendix. 
% The algorithms for the cycle appears in Appendix~\ref{app:cycle}. 
% The pseudocode of our algorithm for the ladder appears in Appendix~\ref{app:static}~and~\ref{app:dynamic}, while proofs appear in Appendix~\ref{app:proofs}.
% More details about the algorithm and analysis for general demand graphs appear in Appendix~\ref{app:general}.

\begin{comment}
We already mentioned that the calculation of the bandwidth in general is NP-hard. In addition, a task of embedding a tree onto a graph in general is also NP-hard, as, for example, an embedding of a tree onto an infinite grid. However, in the case of a $2 \times n$ grid we provide polynomial algorithms: the calculation of a bandwidth and the embedding of a tree onto the $2 \times n$ grid.
\end{comment}

\vspace{-0.3cm}
\section{Model and Background}
\label{sec:model}
\vspace{-0.2cm}

Let us introduce the notation that we are going to use throughout the paper.
% \item 
Let $V(H)$ and $E(H)$ be the sets of vertices and edges in graph $H$, respectively. Sometimes, we just use $V$ and $E$ if the graph $H$ is obvious from the context.
% \item 
Let $d_H(u, v)$ be the distance between $u$ and $v$ in graph $H$.

Let $N$ be the network topology and $\sigma$ be a sequence of pairwise communication requests between nodes in $N$.
Let the demand graph $G$ be the graph built over the nodes in $N$ and the pairs of nodes that appear in $\sigma$, i.e. $G = (V(N), \{\sigma_i = (s_i, d_i) \,|\, \sigma_i \in \sigma\})$. 
We assume that the demand graph is of a certain type and our overall goal will be to embed the demand graph $G$ onto the actual network topology $N$ at a minimum cost.
This is non-trivial as requests are selected from $G$ by an online adversary and $G$ is not known in advance.
In the following, we formalize demand graph embedding and topology reconfiguration.

% \begin{itemize}
% \item 
A configuration (or an embedding) of $G$ (the demand graph) in a graph $N$ (the host network) is an injection of $V(G)$
into $V(N)$; $C_{G \rightarrow N}$ denotes the set of all such configurations.
% \item 
A configuration $c \in C_{G \rightarrow N}$ is said to serve a communication request $(u, v) \in E(G)$ at the cost $d_N(c(u), c(v))$.
% \item 
A finite communication sequence $\sigma = (\sigma_1, \ldots, \sigma_m)$ is served by a sequence of configurations $c_0, c_1, \ldots, c_m \in C_{G \rightarrow N}$.
% \item 
The cost of serving $\sigma$ is the sum of serving each $\sigma_i$
in $c_i$ plus the reconfiguration cost between subsequent configurations $c_i$
and $c_{i+1}$.
% \item 
The reconfiguration cost between $c_i$ and $c_{i + 1}$ is the number of migrations necessary to change from $c_i$ to $c_{i+1}$; a migration swaps the images of two neighbouring nodes $u$ and $v$ under $c$ in $N$.
% \item 
Moreover, $E_i = \{\sigma_1, \ldots, \sigma_i\}$ denotes the first $i$ requests of $\sigma$ interpreted as a set of
edges on $V$.
%, and $D(\sigma) = (V, E_m)$ denotes the demand graph of $\sigma$.
% \end{itemize}
%
We present algorithms for an online self-adjusting linear network:
a network whose topology forms a 1-dimensional grid, i.e., a line.

\vspace{-0.2cm}
\begin{definition}[Working Model]
Let $G$ be the demand graph, $n$ be the number of vertices in  $G$, $N = (\{1, \ldots, n\}, \{(1, 2),(2, 3), \ldots,(n - 1,  n)\}$ be a line (or list) graph $L_n$ (host network), $c$ be a configuration from $C_{G \rightarrow N}$, and $\sigma$ be a sequence of communication requests. The cost of serving $\sigma_i = (u, v) \in \sigma$ is given by $|c(u) - c(v)|$, i.e., the distance between $u$ and $v$ in $N$. Migrations can occur before serving a request and can only occur between nodes configured on adjacent vertices in $N$.
\end{definition}
\vspace{-0.2cm}

In the following we introduce notions relevant to our new results.
% \vspace{-0.2cm}
% \subsection{Preliminaries}
% \vspace{-0.2cm}
% We lay the foundations for presenting our contributions. % (recall the definitions in Section \ref{sec:model}).
% also presented in the model
% \begin{definition}
% An \emph{embedding} of a graph $G$ into graph $N$ is an injective mapping $\varphi: V(G) \rightarrow V(N)$. The set of all embeddings of $G$ into $N$ is denoted as $C_{G \rightarrow N}$.
% \end{definition}
\begin{definition}
A \emph{correct embedding} of a graph $G$ into graph $N$ is an injective mapping $\varphi: V(G) \rightarrow V(N)$ that preserves edges, i.e. 
\begin{align*}
\begin{cases}
    \forall u, v \in V(G) \text{ with } u \neq v \Rightarrow \varphi(u) \neq \varphi(v)\\
    % \forall 
    (u, v) \in E(G) \Rightarrow (\varphi(u), \varphi(v)) \in E(N)
\end{cases}
\end{align*}
\end{definition}

\begin{definition}[Bandwidth]
Given a graph $G$, the \emph{Bandwidth} of an embedding $c \in C_{G \rightarrow L_n}$ is equal to the maximum over all edges $(u, v) \in E$ of $|c(u) - c(v)|$, i.e., the distance between $u$ and $v$ on $L_n$. 
$\bandwidth(G)$ is the minimum bandwidth over all embeddings from $C_{G \rightarrow L_n}$.
\end{definition}

\vspace{-0.2cm}

\begin{remark2}
The $\bandwidth$ computation of an arbitrary graph is an NP-hard problem~\cite{diaz2002survey}.
\end{remark2}

\vspace{-0.2cm}

To save the space, we typically omit the proofs of lemmas and theorems in this paper and put them in Appendix~\ref{app:proofs}.
Here we define the $2\times n$ grid or ladder graph for which we get the main results of our paper.
% \vspace{-0.2cm}
\begin{definition}
A graph $Ladder_n = (V, E)$ is represented as follows. The vertices $V$ are the nodes of the grid $2 \times n$~--- $\{(1, 1), (1, 2), \ldots, (1, n), (2, 1), (2, 2), \ldots, (2, n)\}$. There is an edge between vertices $(x_1, y_1)$ and $(x_2, y_2)$ iff $|x_1 - x_2| + |y_1 - y_2| = 1$.
\end{definition}

%\vspace{-0.3cm}
\noindent
\begin{minipage}{0.6\textwidth}
\begin{lemma}\label{lem:gridbandwidth}
$\bandwidth(Ladder_n) = 2$.
\end{lemma}
\begin{proof}
% \label{lem:gridbandwidth}
%[Lemma \ref{lem:grid bandwidth}]
%\label{proof:lem:grid bandwidth}
The bandwidth is greater than 1, because there are nodes of degree three. The bandwidth of 2 can be achieved via the ``level-by-level'' enumeration as shown on the figure.

\end{proof}
\end{minipage}
\hspace{0.01\textwidth}
\begin{minipage}{0.39\textwidth}
    \centering
    \vspace{0.2cm}
    \includegraphics[scale=0.2]{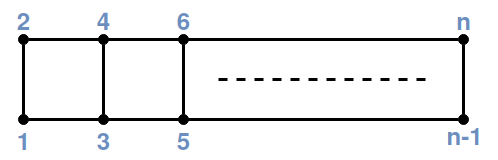}
    \captionof{figure}{Optimal ladder numeration.}
    \label{fig:ladder_numeration}
\end{minipage}

\vspace{-0.2cm}
\begin{lemma}\label{lem:embedding enumeration}
For each subgraph $S$ of a graph $G$, $\bandwidth(S) \leq \bandwidth(G)$.
\end{lemma}
% \begin{proof}
% See \proofref{lem:embedding enumeration}.
% \end{proof}

\vspace{-0.5cm}
\subsection{Background}
\label{sec:background}
\vspace{-0.1cm}

Let us overview the previous results from~\cite{avin2019self}.
In that work, both the demand and the host graph (network topology) were the line graph $L_n$ on $n$ vertices.
% In that work, the demand graph was the line graph $L_n$, as the network graph. 
It was shown that there exists an algorithm that performs $O(n^2 \log n)$ migrations in total, while serving the requests themselves in $O(1)$. By that, if the number of requests is $\Omega(n^2 \log n)$ then each request has $O(1)$ amortized cost.
%Moreover, $n^2 \log n$ is the lower bound on the total cost: if there are $\Theta(n^2)$ requests then the total cost of any online algorithm must be $\Omega(n^2 \log n)$, meaning that the static-optimality factor is $\log n$.

% \vspace{-0.2cm}
\begin{theorem}[Avin et al. \cite{avin2019self}]
\label{thm:line}
Consider a linear network $L_n$ and a linear demand graph. % $D(\sigma)$ where $\sigma$ is the sequence of requests.
There is an algorithm such that the total time spent on migrations is $O(n^2 \log n)$, while each request is performed in $O(1)$ omitting the migrations.
\end{theorem}
\vspace{-0.1cm}

We give an overview of this algorithm. At each moment in time, we know some subgraph of the line demand graph. For each new communication request, there are two cases: 1)~the edge from the demand graph is already known~--- then, we do nothing; 2) the new edge is revealed. In the second case, this edge connects two connected components. We just move the smallest component on the line network closer to the largest component. The move of each node in one reconfiguration does not exceed $n$. Since, the total number of reconfigurations in which the node participates does not exceed $\log n$, we have $O(n^2 \log n)$ upper bound on the algorithm. From \cite{avin2019self}, $\Omega(n^2 \log n)$ is also the lower bound on the total cost. Thus, the algorithm is asymptotically optimal in the terms of complexity. 

% \vspace{-0.2cm}
\begin{corollary}
If $|\sigma| = \Omega(n^2 \log n)$ the amortized service cost per request is $O(1)$.
\end{corollary}
% \vspace{-0.2cm}

The algorithms are not obliged to perform migrations at all, but the sum of costs for $\Theta(n^2)$ requests can be lower-bounded with $\Omega(n^2 \log n)$.

% \vspace{-0.2cm}
\begin{theorem}[Lower bound, Avin et al. \cite{avin2019self}]
For every online algorithm $ON$ there is a sequence of requests $\sigma_{ON}$ of length $\Theta(n^2)$ with the demand graph being a line, such that $cost(ON(\sigma_{ON})) = \Omega(n^2 \log n)$.
\end{theorem}
\vspace{-0.2cm}

That implies $\Omega(\log n)$ optimality factor since any offline algorithm knowing the whole request sequence $\sigma$ in advance can simply reconfigure the network to match the (line) demand graph by paying $\Theta(n^2)$ in the worst case.

\vspace{-0.2cm}
\section{Summary of contributions}
\label{sec:results}
\vspace{-0.2cm}

In this work we present self-adjusting networks with a line topology for a demand graph that is either a cycle, or a $2\times n$ grid (ladder), or an arbitrary graph.
We study offline and online algorithms on how to best embed the demand graph on the line graph, such that the embedding cost is minimized.
The online case is more challenging, as the demand graph is revealed edge-by-edge and the embedding changes, with a cost.
% In this work, we present self-adjusting linear networks with demand graphs different from the line, which was studied in~\cite{avin2019self}. 
% We present results for generalizations of the line graph: the cycle and the ladder, and we finish with the general result for arbitrary graphs. 
The result for the cycle follows from~\cite{avin2019self} almost directly. However the result for the ladder is non-trivial and
requires new techniques; 
% very technical~--- 
it is not simple to reconfigure a subgraph on a $2 \times n$ grid after revealing a new edge in order to get $O(n^2 \log n)$ cost of modifications in total.
We give an overview of each case below.

\vspace{-0.2cm}
\subsection{Cycle demand graph}
\vspace{-0.2cm}

% We start with the most simple generalization result~--- when the demand graph is the cycle on $n$ vertices. 
We 
% have to 
start with the following observation.
%
% \begin{observation}
Let $C_n$ be a cycle graph on $n$ vertices, i.e., $E(C_n) = \{(1, 2), \ldots, (n - 1, n), (n, 1)\}$. Then, $\bandwidth(C_n) = 2$.
% \end{observation}
We give a brief description of how the algorithm works. We start with the same algorithm as for the line (Section \ref{sec:background}): while the number of revealed edges is not more than $n-1$, we can emulate the algorithm for the line. When the last edge appears we restructure the whole embedding in order to get bandwidth $2$, which is the cycle bandwidth.
% , as cycle has. 
For the whole restructuring using swaps, we pay no more than $O(n^2)$. This cost is less than the total time spent on the reconstruction $\Omega(n^2 \log n)$.

\vspace{-0.2cm}
\begin{theorem}\label{thm:cycle}
Suppose the demand graph is $C_n$. There is an algorithm such that the total cost spent on the migrations is $O(n^2 \log n)$ and each request is performed in $O(1)$.
In particular, if the number of requests is $\Omega(n^2 \log n)$ each request has $O(1)$ amortized cost.
\end{theorem}

\vspace{-0.2cm}
The full proof appears in Appendix~\ref{app:cycle}.
\vspace{-0.2cm}
\begin{remark2}
% Please note that 
The lower bound with $\Omega(n^2 \log n)$ that was presented for a line demand graph still holds in the case of a cycle, since the cycle contains the line as the subgraph. Thus, our algorithm is optimal.
\end{remark2}

\vspace{-0.5cm}
\subsection{Ladder demand graph}
\vspace{-0.2cm}

Now, we state the main result of the paper~--- the algorithm for the case when the demand graph is an ladder.
\begin{theorem}\label{thm:grid request graph}
Suppose a demand graph is an ladder. There is an algorithm such that the total cost spent on the migrations is $O(n^2 \log n)$ and each request is performed in $O(1)$.
In particular, if the number of requests is $\Omega(n^2 \log n)$ each request has $O(1)$ amortized cost.
\end{theorem}

We provide a brief description of the algorithm.
%
% At first, 
We say that an ladder has $n$ levels from left to right: i.e., the nodes $(1, x)$ and $(2, x)$ are on the same level $x$ (see Figure~\ref{fig:ladder_numeration}).
On a high-level, we use the same algorithmic approach as in Theorem~\ref{thm:line} for the line demand graph. The main difference is that instead of embedding the demand graph right away onto the line network, at first, we ``quasi-embed'' the graph onto the $2n$-ladder graph, which then we embed onto the line.
By ``quasi-embedding'' we mean a relaxation of the embedding defined earlier: at most \textbf{three} vertices of the demand graph are mapped on each level of the ladder.

Suppose for a moment that we have a dynamic algorithm that quasi-embeds the graph onto the $2n$-ladder.
Given this quasi-embedding we can then embed the $2n$-ladder onto the line $L_n$. We consequently go through from level $1$ to level $2n$ of our ladder and map (at most three) vertices from the level to the line in some order (see Theorem~\ref{lem:gridbandwidth}). Such a transformation from the ladder to the line costs only a constant multiplication overhead.

We explain briefly how to design a dynamic quasi-embedding algorithm with the desired complexity.
At first, we present a static quasi-embedding algorithm, i.e., we are given a subgraph of the ladder and we need to quasi-embed it. This algorithm consists of three parts: embed a tree, embed a cycle, embed everything together. To embed a tree we find a special path in it, named trunk. We embed this trunk from left to right: one vertex per level. All the subgraphs connected to trunk are pretty simple and can be easily quasi-embedded in parallel to the trunk (see Figure~\ref{fig:quasi-embedding-lite}).
To embed a cycle we just have to decide which orientation it should have. To simplify the algorithm we embed only the cycles of length at least $6$, omitting the cycles of length $4$. This decision increases the multiplicative constant of the cost. Finally, we embed the whole graph: we construct its cycle-tree decomposition and embed cycles and trees one by one from left to right.

Now, we give a high-level description of our dynamic algorithm. We maintain the invariant that all the components are quasi-embedded. When an already served request appears, we do nothing. The complication comes from a newly revealed edge-request. There are two cases. The first one is when the edge connects nodes in the same component~--- thus, there is a cycle. We redo only the part of the quasi-embedding of the component around the new cycle; the rest of the component remains. In the second case, the edge connects two components. We move the smaller component to the bigger one as in Theorem~\ref{thm:line}. The bigger component does not move and we redo the quasi-embedding of the smaller one.

\begin{wrapfigure}[6]{r}{0.6\textwidth}
\vspace{-1cm}
  \begin{center}
    \includegraphics[scale=0.3]{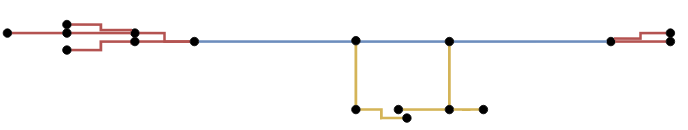}
  \end{center}
%   \vspace{-0.2cm}
\caption{Quasi-correct embedding of a tree}
\label{fig:quasi-embedding-lite}
\vspace{2cm}
\end{wrapfigure}
Now, we briefly calculate the complexity of the dynamic algorithm. For the requests of the first case, if the nodes are on the cycle for the first time (this event happens only once for each node), we pay $O(n)$ for it. 
% Otherwise, there are nodes on the new cycle that already are on the cycle. 
Otherwise, there are already nodes in the cycle.
In this case we make sure to re-embed the existing cycle in a way that all the nodes are moved for a $O(1)$ distance. As for the neighboring nodes, it can be shown that each node is moved only once as a part of the cycle neighborhood, so we also bound this movement with $O(n)$ cost. This gives us $O(n^2)$ complexity in total~--- each node is moved by at most $O(n)$. For the requests of the second case, we always move the smaller component and, thus, we pay $O(n^2 \log n)$ in total: each node can be moved by $O(n)$ at most $O(\log n)$ times, i.e., any node can be at most $\log n$ times in the ``smaller'' component. Our algorithm matches the lower bound, since the ladder contains $L_n$ as a subgraph.

% We hope that we gave an intuition of the algorithm. 
% We present all the necessary notions in the next section, then, in Section~\ref{sec:algorithm} we present a more detailed overview of the algorithm and, finally, we calculate its cost in Section~\ref{sec:cost}.

\begin{comment}
\begin{itemize}
\item Quasi-embedding. Embed into levels, three in a row.
\item Invariants: Hand and foot cannot be on one level (foot invariant), trunk core maps into ``ones''.
\item Introduce trees on the grid.
\item Add edge in one component and get cycle. Take frame off and embed it and relocate the first hands up and down. Get quasi-embedding. We neglect cycle 4.
\item Add edge between components. The smaller one to the bigger one. Cycle, trunkCore, inner line-graph, exit-graphs. Line-graph: enlargement, trunkCore: is fine (foot invariant), cycle: restruct and embed simply (any graph embedding), exit-graphs 
\item 
\end{itemize}
\end{comment}

\vspace{-0.2cm}
\subsection{General graph}
% \vspace{-0.2cm}
We finish the list of contributions with a general result; the case where the demand graph is an arbitrary graph $G$.
The full proofs are available at Appendix~\ref{app:general}. 

% \begin{theorem}
% Given a graph $G$. There is an algorithm such that the total time spent on the reconfigurations is $O(|E| \cdot |V|^2)$ and each request is performed in $\bandwidth(G)$.

% In other words, if the number of requests is $\Omega(|E| \cdot |V|^2)$ each request costs $O(\bandwidth(G))$ amortized time.
% \end{theorem}

\vspace{-0.1cm}
\begin{theorem}\label{thm:arbitrary-graph}
Suppose we are given a (demand) graph $G$ and an algorithm $B$, that for any subgraph $S$ of $G$ outputs an embedding $c \in C_{S \rightarrow L_{|V(G)|}}$ with  bandwidth less than or equal to $\lambda \cdot \bandwidth(G)$ for some $\lambda$. 
Then, for any sequence of requests $\sigma$ with a demand graph $G$ there is an algorithm that serves $\sigma$ with a total cost of $O(|E(G)| \cdot |V(G)|^2 + \lambda \cdot \bandwidth(G) \cdot |\sigma|)$.
In particular, if the number of requests is $\Omega(|E(G)| \cdot |V(G)|^2)$ each request has $O(\lambda \cdot \bandwidth(G))$ amortized cost.
\end{theorem}
\vspace{-0.1cm}

Here we give a brief description of the algorithm. Suppose that the current configuration $c_i$ is the embedding of the current demand graph $G_i$ onto $L_{|V(G)|}$ after $i$ requests. Now, we need to serve a new request in $\lambda \cdot \bandwidth(G_i) \leq \lambda \cdot \bandwidth(G)$. If the corresponding edge already exists in the demand graph, we simply serve the request without the reconfiguration.
Now, suppose the request reveals a new edge and we get the demand graph $G_{i+1}$. Using the algorithm $B$ we get the configuration (embedding) $c_{i+1}$ that has $\lambda \cdot \bandwidth(G_{i+1}) \leq \lambda \cdot \bandwidth(G)$. To serve the request fast, we should rebuild the configuration $c_i$ into the configuration $c_{i+1}$. By using the swap operations on the line we can get from $c_i$ to $c_{i+1}$ in $O(|V(G)|^2)$ operations: each vertex moves by at most $V(G)$. After the reconfiguration we can serve the request with the desired cost.

A new edge appears at most $|E(G)|$ times while the reconfiguration costs $|V(G)|^2$. Each request is served in $\lambda \cdot \bandwidth(G)$. Thus, the total cost of requests $\sigma$ is $O(|E(G)| \cdot |V(G)|^2 + \lambda \cdot \bandwidth(G) \cdot |\sigma|)$.

%The The algorithm works as follows: after each new revealed edge we rebuild the configuration on the line graph in order to match the bandwidth. Each vertex will move $\Omega(|E(G)|)$ times by at most than $O(|V(G)|)$. Thus, the total additional cost is $O(|E(G)| \cdot |V(G)|^2)$.

\vspace{-0.1cm}
\begin{lemma}
\label{lem:lower-bound:arbitrary-graph}
Given a demand graph $G$. For each online algorithm $ON$ there is a request sequence $\sigma_{ON}$ such that $ON$ serves each request from $\sigma_{ON}$ for a cost of at least $\bandwidth(G)$.
\end{lemma}

%\begin{remark2}
%Using the previous lemma, $\lambda$ is a static-optimality factor for the case when $|\sigma| = \Omega(|E(G)| \cdot |V(G)|^2)$.
%\end{remark2}

\vspace{-0.2cm}
% \section{Ladder: notation}
% \label{sec:notions}
% \vspace{-0.2cm}

% \vspace{-0.2cm}
% \subsection{Tree embedding}
% \vspace{-0.2cm}

% \vspace{-0.2cm}
% \subsection{Embedding with Cycles}
% \vspace{-0.2cm}

\vspace{-0.2cm}
% \section{Ladder: algorithm}
\section{Embedding a ladder demand graph}
\label{sec:algorithm}

% \vspace{-0.3cm}
We present our algorithms for embedding a demand graph that is a subgraph of the ladder graph ($2\times n$-grid) on the line graph. We first present the offline case, where the demand graph is known in advance (Section \ref{sec:static}). Then we present the dynamic case, where requests are revealed online, revealing also the demand graph and thus possibly changing the current embedding (Section \ref{sec:dynamic}).
Finally, we discuss the cost of the dynamic case in Section \ref{sec:cost}.

Though our final goal is to embed a demand graph into the line, we will first focus on how to embed a partially-known demand graph into $Ladder_N$, where $N$ is large enough to make the embedding possible, i.e., not more then $2n$. When we have such an embedding one might construct an embedding from $Ladder_N$ into $Line_n$, simply composing it with a level by level (see the proof of Lemma \ref{lem:gridbandwidth}) embedding of $Ladder_N$ to $Line_{2N}$ and then by omitting empty images we get $Line_n$. 
Such a mapping of $Ladder_N$ to $Line_{2N}$ enlarges the bandwidth for at most a factor of 2,  
%
% Before going into details, we have to note that we will embed the graph not onto $Ladder_n$, but onto $Ladder_{N}$, where $N$ is large but does not exceed $2 \cdot n$. This affects the bandwidth of the graph with only a constant factor and the map of a graph from the grid onto the line
but significantly simplifies the construction of our embedding.

% \begin{definition}
% A line-graph on n vertices is a graph with $V = [1, \ldots, n]$ and $E = \{(i, i + 1) \mid i \in [1, \ldots, n - 1]\}$.

% We refer to the $i$-th node of a line-graph $l$ as $l[i]$.
% \end{definition}

\vspace{-0.1cm}

\begin{definition}
An ladder graph $l$ consists of two line-graphs on $n$ vertices $l_1$ and $l_2$ with additional edges between the lines: $\{(l_1[i], l_2[i]) \mid i \in [n]\}$, where $l_j[i]$ is the $i$-th node of the line-graph $l_j$.
%
%Further, we denote $2 \times n$ grid as $Ladder_n$.
%
We call the set of two vertices, $\{l_1[i], l_2[i]\}$, the $i$-th level of the ladder and denote it as $level_{Ladder_n}(i)$ or just $level(i)$ if it is clear from the context. We refer to $l_1[i]$ and $l_2[i]$ as $level(i)[1]$ and $level(i)[2]$, respectively.
We say that $level\langle v \rangle = i$ for $v \in V(Ladder_n)$ if $v \in level_{Ladder_n}(i)$.
We refer to $l_1$ and $l_2$ as the sides of the ladder.
%
%We refer to the side of $v \in V(Ladder_n)$ with $side(v)$. 
%$side(v) = 
%\begin{cases}
%    1,\ v \in V(l_1)\\
%    2, \text{ otherwise}
%\end{cases}$
%We refer to the other node of the level of $v$ with $opposite(v)$. 
\end{definition}

\vspace{-0.3cm}
\subsection{Static quasi-embedding}
\label{sec:static}
% \vspace{-0.1cm}

We start with one of the basic algorithms~--- how to quasi-embed on $Ladder_N$ with large $N$ any graph that can be embedded in $Ladder_n$. We present a tree and cycle embedding and then we show how to to combine them in embedding a general component (by first doing a cycle-tree decomposition). The whole algorithm is presented in Appendix~\ref{app:static}.

\vspace{-0.4cm}
\subsubsection{Tree embedding}
In this case, our task is to embed a tree on a ladder graph. We start with some definitions and basic lemmas.

\vspace{-0.1cm}
\begin{definition}
\label{def:trunk}
Consider some correct embedding $\varphi$ of a tree $T$ into $Ladder_n$. Let $r = \arg\max\limits_{v \in V(T)} level\langle \varphi(v)\rangle$ and
$l = \arg\min\limits_{v \in V(T)} level\langle \varphi(v)\rangle $ be the ``rightmost'' and ``leftmost'' nodes of the embedding, respectively. The trunk of $T$ is a path in $T$ connecting $l$ and $r$. The trunk of a tree $T$ for the embedding $\varphi$ is denoted with $trunk_\varphi(T)$.
\end{definition}

\vspace{-0.2cm}
\begin{definition}
\label{def:occupied}
Let $T$ be a tree and $\varphi$ be its correct embedding into $Ladder_n$.
The level $i$ of $Ladder_n$ is called \emph{occupied} if there is a vertex $v \in V(T)$ on that level, i.e., $\varphi(v) \in level_{Ladder_n}(i)$.
\end{definition}

\vspace{-0.4cm}
\begin{statement}\label{trunk each level}
For every occupied level $i$ there is $v \in trunk_\varphi(T)$ such that $v \in level(i)$. 
\end{statement}
% \begin{qedproof2}
\begin{proof}
By the definition of the trunk, an image goes from the minimal occupied level to the maximal. It cannot skip a level since the trunk is connected and the correct embedding preserves connectivity.
% \end{qedproof2}
\end{proof}

The trunk of a tree in an embedding is a useful concept to define since the following hold for it. 
The proofs for the lemmas in this section appear in Appendix \ref{app:proofs}.

\begin{lemma}\label{lem:left line-graphs}
Let $T$ be a tree correctly embedded into $Ladder_n$ by some embedding $\varphi$. Then, all the connected components in $T \,\setminus\, trunk_\varphi(T)$ are line-graphs.  
\end{lemma}
% \begin{proof}
% See \proofref{lem:left line-graphs}.
% \end{proof}

\begin{lemma}\label{trunk certain nodes}
For the tree $T$ and for each node $v$ of degree three (except for maximum two of them) we can verify in polynomial time if for any correct embedding $\varphi$, $trunk_\varphi(T)$ passes through $v$ or not. 
\end{lemma}
% \begin{proof}
% See \proofref{trunk certain nodes}.
% \end{proof}

% \begin{definition}
\emph{Support nodes} are the nodes of two types: either a node of degree three without neighbours of degree three or a node that is located on some path between two nodes with degree three.
The path through passing through all support nodes is called \emph{trunk core}. We denote this path for a tree $T$ as $trunkCore(T)$.
%Note that it can be embedded into $Ladder_n$. \ap{Didn't get the last sentence}
%
Intuitively, the trunk core consists of vertices that lie on a trunk of any embedding.
% \end{definition}
%
% \begin{remark2}
It can be proven that the support nodes appear in the trunk of every correct embedding (proof appears in the appendix).
% \end{remark2}

\begin{definition}
\label{def:simple-graphs}
Let $T$ be a tree. All the connected components in $T\,\setminus\,trunkCore(T)$ are called \emph{simple-graphs} of tree $T$.
\end{definition}

\begin{lemma}\label{simple-graphs line-graphs}
The simple-graphs of a tree $T$ are line-graphs. 
\end{lemma}
% \begin{proof}
% See \proofref{simple-graphs line-graphs}.
% \end{proof}

\noindent
\begin{minipage}{0.7\textwidth}
\begin{definition}
\label{def:hands}
The edge between a simple-graph and the trunk core is called a \emph{leg}.
The end of a leg in the simple-graph is called a \emph{head} of the simple-graph.
The end of a leg in the trunk core is called a \emph{foot} of the simple-graph.

If you remove the head of a simple-graph and it falls apart into two connected components, such simple-graph is called \emph{two-handed} and those parts are called its \emph{hands}. Otherwise, the graph is called \emph{one-handed}, and the sole remaining component is called a \emph{hand}. If there are no nodes in the simple-graph but just a head it is called zero-handed.
\end{definition}
\end{minipage}
\hspace{0.01\textwidth}
\begin{minipage}{0.29\textwidth}
\centering
\includegraphics[width=\textwidth]{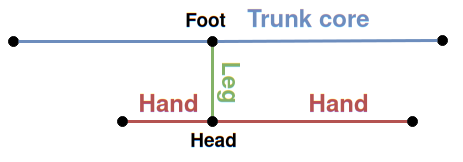}
\captionof{figure}{Hands, Legs, and Trunk core.}
\end{minipage}

\vspace{-0.1cm}

\begin{definition}
\label{def:exit}
\label{def:inner}
A simple-graph connected to some end node of the trunk core is called \emph{exit-graph}.
% \end{definition}
% \vspace{-0.2cm}
% \begin{definition}
% \label{def:inner}
A simple-graph connected to an inner node of the trunk core is called \emph{inner-graph}.
\end{definition}

\vspace{-0.2cm}

Please note that the next definition is about a much larger ladder graph, $Ladder_N$, rather than $Ladder_n$. Here, $N$ is equal to $2n$ to make sure that we have enough space to embed.
\begin{definition}
\label{def:quasi-correct}
An embedding $\varphi: V(G) \rightarrow V(Ladder_N)$ of a graph $G$ into $Ladder_N$ is called \emph{quasi-correct} if:
\vspace{-0.2cm}
\begin{itemize}
    \item $(u, v) \in E(G) \Rightarrow (\varphi(u), \varphi(v)) \in E(Ladder_N)$, i.e., images of adjacent vertices in $G$ are adjacent in the grid.
    \item There are no more than \textbf{three} nodes mapped into each level of $Ladder_N$, i.e., the two grid nodes on each level are the images of no more than three nodes.
\end{itemize}
\end{definition}

\vspace{-0.2cm}

We can think of a quasi-correct embedding as an embedding into levels of the grid with no more than three nodes embedded to the same level. Then, we can compose this embedding with an embedding of a grid into the line which is the enumeration level by level. More formally if a node $u$ is embedded to level $i$ and a node $v$ is embedded to level $j$ and $i < j$ then the resulting number of $u$ on the line is smaller than the number of $v$, but if two nodes are embedded to the same level, we give no guarantee.

\vspace{-0.2cm}
\begin{lemma}\label{quasi const cost}
Any graph mapped into the ladder graph by the quasi-correct embedding described above can be mapped onto the line level by level with the property that any pair of adjacent nodes are embedded at the distance of at most five.
\end{lemma}
% \begin{proof}
% See \proofref{quasi const cost}.
% \end{proof}

\vspace{-0.2cm}
Assume, we are given a tree $T$ that can be embedded into $Ladder_n$. Furthermore, there are two special nodes in the tree: one is marked as \textit{R} (right) and another one is marked as \textit{L} (left). It is known that there exists a correct embedding of $T$ into $Ladder_n$ with \textit{R} being the rightmost node, meaning no node is embedded more to the right or to the same level, and \textit{L} being the leftmost node.

We now describe how to obtain a quasi-correct embedding of $T$ onto $Ladder_N$ with \textit{R} being the rightmost node and \textit{L} being the leftmost one while $L$ is mapped to $ImageL$~--- some node of the $Ladder_N$. Moreover, our embedding obeys the following invariant.

\vspace{-0.2cm}
\begin{invariant}[Septum invariant]
\label{invariant:septum}
For each inner simple-graph, its foot and its head are embedded to the same level and no other node is embedded to that level.
\end{invariant}

\vspace{-0.2cm}
\begin{wrapfigure}{r}{0.5\textwidth}
\vspace{-1cm}
  \begin{center}
    \includegraphics[scale=0.25]{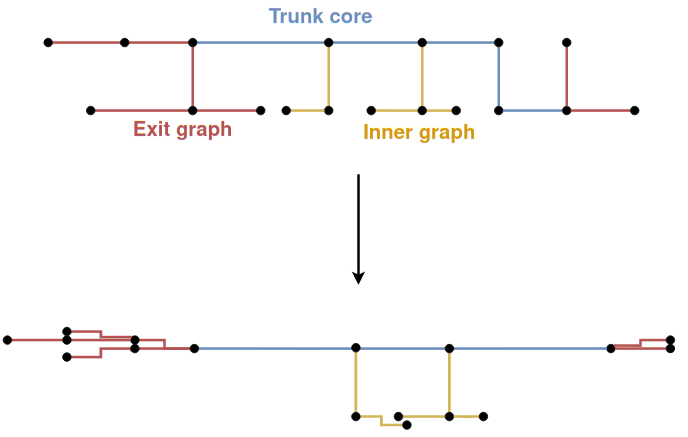}
  \end{center}
%   \vspace{-0.2cm}
\caption{Example of a quasi-correct embedding}
\label{fig:quasi-embedding:example}
\vspace{-1cm}
\end{wrapfigure}
We embed a path between $L$ and $R$ simply horizontally and then we orient line-graphs connected to it in a way that they do not violate our desired invariant. It can be shown that it is always possible if $T$ can be embedded onto $Ladder_n$. The pseudocode is in Appendix Algorithm~\ref{alg:LR:tree:embedding}.

Suppose now that not all information, such as $R$, $L$, and $ImageL$, is provided. We explain how we can embed a tree $T$.
%In the case the variable is not given we denote its value with $\None$.
%
We first get the \emph{trunk core} of the given tree. This can be done by following the definition. 
Now the idea would be to first embed the trunk core and its inner line-graphs using a tree embedding presented earlier with $R$ and $L$ to be the ends of the trunk core. Then, we embed exit-graphs strictly horizontally ``away'' from the trunk core. That means, that the hands of exit-graphs that are connected to the right of the trunk core are embedded to the right, and the hands of those exit-graphs that are connected to the left of the trunk core are embedded to the left. An example of the quasi-correct embedding is shown in Figure~\ref{fig:quasi-embedding:example}.

% \begin{figure}
% \begin{center}
%     \includegraphics[scale=0.3]{images/quasi-embedding.png}
% \end{center}
% \caption{Example of a quasi-correct embedding}
% \label{fig:quasi-embedding:example}
% \end{figure}

If a tree does not have a trunk core, then its structure is quite simple (in particular it has no more than two nodes of degree three). Such a tree can be embedded without conflicts.
The pseudocode appears in Appendix Algorithm~\ref{alg:tree:embedding}.

\vspace{-0.3cm}
\subsubsection{Cycle embedding}
Now, we show how to embed a cycle into $Ladder_N$. First, we give some important definitions and lemmas.

\vspace{-0.2cm}
\begin{definition}
\label{def:maximal-cycle}
A \emph{maximal} cycle $C$ of a graph $G$ is a cycle in $G$ that cannot be enlarged, i.e., there is no other cycle $C'$ in $G$ such that $V(C) \subsetneq V(C')$.
\end{definition}

\vspace{-0.2cm}

\noindent
\begin{minipage}{0.69\linewidth}
\begin{definition}
\label{def:whiskers}
Consider a graph $G$ and a maximal cycle $C$ of $G$. A whisker $W$ of $C$ is a line graph inside $G$ such that:
1) $V(W) \neq \emptyset$ and $V(W) \cap V(C) = \emptyset$.
2) There exists only one edge between the cycle and the whisker $(w, c)$ for $w \in V(W)$ and $c \in V(C)$. Such $c$ is called a \emph{foot} of $W$. The nodes of $W$ are enumerated starting from $w$.
%    \item There are no nodes of degree three in $V(W)$.
3) $W$ is maximal, i.e., there is no $W'$ in $G$ such that $W'$ satisfies previous properties and $V(W) \subsetneq V(W')$.
\end{definition}
\end{minipage}
\hspace{0.01\textwidth}
\begin{minipage}{0.3\linewidth}
\centering
    \includegraphics[width=\textwidth]{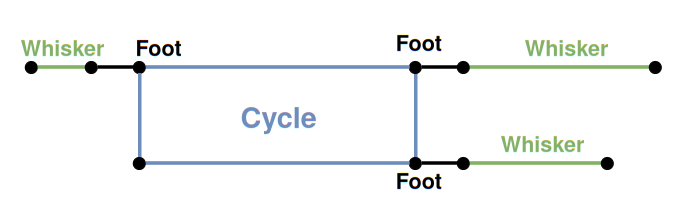}
    \captionof{figure}{Cycle and its Whiskers.}
\end{minipage}

\va{Redraw the picture with whisker and foot with the black edge.}

\vspace{-0.2cm}

\begin{definition}
\label{def:adjacent-whiskers}
Suppose we have a graph $G$ that can be correctly embedded into $Ladder_n$ by $\varphi$ and a cycle $C$ in $G$. Whiskers $W_1$ and $W_2$ of $C$ are called adjacent (or neighboring) for the embedding $\varphi$ if 
% \begin{align*}
    $\forall i \leq \min(|V(W_1)|, |V(W_2)|)\ (\varphi(W_1[i])$, $\varphi(W_2[i])) \in E(Ladder_n)$.
% \end{align*}

% maybe this one?
% \begin{align*}
%     \forall i \in \mathrm{argmin}_{\{V(W_i) \,|\, i\in\{1,2\}\}} |V(W_i)|,  (\varphi(W_1[i]), \varphi(W_2[i])) \in E(Ladder_n)
% \end{align*}
\end{definition}

\vspace{-0.2cm}

\begin{lemma}\label{adjacent whiskers}
Suppose we have a graph $G$ that can be correctly embedded into $Ladder_n$ and there exists a maximal cycle $C$ in $G$ with at least $6$ vertices with two neighbouring whiskers $W_1$ and $W_2$ of $C$, i.e., $(\text{foot}(W_1), \text{foot}(W_2)) \in E(G)$. Then, $W_1$ and $W_2$ are adjacent in any correct embedding of $G$ into $Ladder_N$.
\end{lemma}
% \begin{proof}
% See \proofref{adjacent whiskers}.
% \end{proof}

\vspace{-0.2cm}

\noindent
\begin{minipage}{0.68\linewidth}
\begin{definition}
\label{def:frame}
Assume we have a graph $G$ and a maximal cycle $C$ of length at least $6$. The frame for $C$ is a subgraph of $G$ induced by vertices of $C$ and $\{W_1[i], W_2[i] \mid i \leq \min(|V(W_1)|, |V(W_2)|)\}$ for each pair of adjacent whiskers $W_1$ and $W_2$.
Adding all the edges $\{(W_1[i], W_2[i]) \mid i \leq \min(|V(W_1)|, |V(W_2)|)\}$ for each pair of adjacent whiskers $W_1$ and $W_2$ makes a frame \emph{completed}. 
\end{definition}
\end{minipage}
\hspace{0.01\textwidth}
\begin{minipage}{0.3\linewidth}
\centering
    \includegraphics[width=\textwidth]{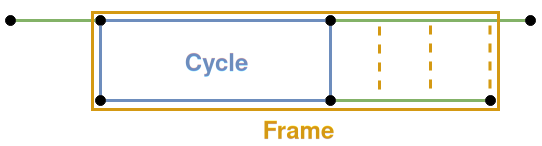}
    \captionof{figure}{Cycle, its frame, and edges (dashed) to make the frame completed}
\end{minipage}

Given a cycle $C$ of length at least six and its special nodes $L, R \in V(C)$, we construct a correct embedding of $C$ into $Ladder_N$ with $level\langle L \rangle \leq level\langle u \rangle \leq level\langle R \rangle\ \forall u \in V(C)$, while $L$ is mapped into the node $ImageL$.

%For convenience we assume that for every node $v$ there is a local consecutive numeration of a cycle starting at $v$. The number of a node $u \in V(C)$ is named $number_v(u)$ while the node number $i$ is $C_v[i]$.

We first check if it is possible to satisfy the given constraints of placing the $L$ node to the left and a $R$ node to the right. If it is indeed possible, we place $L$ to the desired place $ImageL$ and then we choose an orientation (clockwise or counterclockwise) following which we could embed the rest of the nodes, keeping in mind that $R$ must stay on the rightmost level. The pseudocode appears in Appendix Algorithm~\ref{alg:LR:cycle:embedding}.

Now, suppose that not all information, such as $R$, $L$, and $ImageL$, is provided.
We reduce this problem to the case when the missing variables are known. This subtlety might occur since there are inner edges in the cycle. In this case, we choose missing $L/R$ more precisely in order to embed an inner edge vertically. For more intuition, please see Figures \ref{fig:cycle bad} and \ref{fig:cycle good}. A dashed line denotes an inner edge.
The pseudocode appears in the Appendix (Algorithm~\ref{alg:cycle:embedding}).
\begin{figure}[h]
\centering
\begin{subfigure}{.5\textwidth}
  \centering
  \includegraphics[scale=0.15]{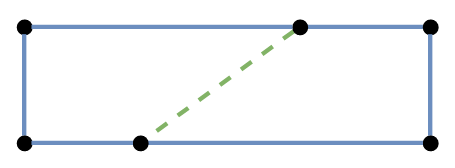}
  \caption{Incorrect cycle embedding}
  \label{fig:cycle bad}
\end{subfigure}%
\begin{subfigure}{.5\textwidth}
  \centering
  \includegraphics[scale=0.15]{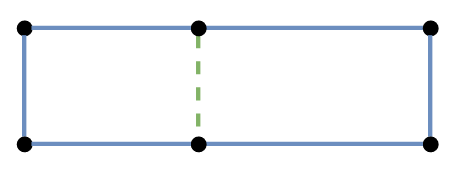}
  \caption{Correct cycle embedding}
  \label{fig:cycle good}
\end{subfigure}
\caption{Cycle embeddings.}
\label{fig:cycleembeddings}
\end{figure}

% \begin{figure}[H]
%   \centering
%   \begin{minipage}[b]{0.25\textwidth}
%     \includegraphics[width=\textwidth]{images/cycle_bad.png}
%     \caption{Incorrect cycle embedding}
%     \label{fig:cycle bad}
%   \end{minipage}
%   \qquad\qquad\qquad
%   \begin{minipage}[b]{0.25\textwidth}
%     \includegraphics[width=\textwidth]{images/cycle_good.png}
%     \caption{Correct cycle embedding}
%     \label{fig:cycle good}
%   \end{minipage}
% \end{figure}

\vspace{-1.3cm}
\subsubsection{Embedding a connected component of the demand graph}
Combining the previous results, we can now explain how to embed onto $Ladder_N$ a connected component $S$ that can be embedded onto $Ladder_n$.
\begin{definition}
By the cycle-tree decomposition of a graph $G$ we mean a set of maximal cycles $\{C_1,\ldots C_n\}$ of $G$ and a set of trees $\{T_1,\ldots,T_m\}$ of $G$ such that 
\begin{itemize}
    \item $\bigcup\limits_{i \in [n]}V(C_i) \cup \bigcup\limits_{i \in [m]}V(T_i) = V(G)$
    \item $V(C_i) \cap V(C_j) = \emptyset\ \forall i \neq j$
    \item $V(T_i) \cap V(T_j) = \emptyset\ \forall i \neq j$
    \item $V(T_i) \cap V(C_j) = \emptyset\ \forall i\in [m],j\in[n]$
    \item $\forall i \neq j\ \forall u \in V(T_i)\ \forall v \in V(T_j)\ (u, v) \notin E(G)$
\end{itemize}
\end{definition}

We start with an algorithm on how to make a cycle-tree decomposition of $S$ assuming no uncompleted frames.
To obtain a cycle-tree decomposition of a graph: 1)~we find a maximal cycle; 2)~we split the graph into two parts by logically removing the cycle; 3)~we proceed recursively on those parts, and, finally, 4)~we combine the results together maintaining the correct order between cycle and two parts (first, the result for one part, then the cycle, and then the result for the second part).
Since we care about the order of the parts, we say that it is a \emph{cycle-tree decomposition chain}. 
The decomposition pseudocode appears in the Appendix Algorithm~\ref{alg:decomposition}.

We describe how to obtain a quasi-correct embedding of $S$. We preprocess $S$: 1)~we remove one edge from cycles of size four; 2)~we complete uncompleted frames with vertical edges. Then, we embed parts of $S$ from the cycle-tree decomposition chain one by one in the relevant order using the corresponding algorithm (either for a cycle or for a tree embedding) making sure parts are glued together correctly. The pseudocode appears in Appendix Algorithm~\ref{alg:component:embedding}.

\vspace{-0.2cm}
\subsection{Online quasi-embedding}
\label{sec:dynamic}
% \vspace{-0.2cm}

In the previous subsection, we presented an algorithm on how to quasi-embed a static graph. Now, we will explain how to operate when the requests are revealed in an online manner. The full version of the algorithm is presented in Appendix~\ref{app:dynamic}.

There are two cases: a known edge is requested or a new edge is revealed. In the first case the algorithm does nothing since we already know how to quasi-correctly embed the current graph and, thus, we already can embed into the line network with constant bandwidth. Thus, further, we will consider only the second case.

We describe how one should change the embedding of the graph after the processing of a request in an online scenario. At each moment some edges of the demand graph $Ladder_n$ are already revealed, forming connected components.
After an edge reveal we should reconfigure the target line graph. For that, instead of line reconfiguration we reconfigure our embedding to $Ladder_N$ that is then embedded to the line level by level and introduces a constant factor. So, we can consider the reconfiguration only of $Ladder_N$ and forget about the target line graph at all.
When doing the reconfiguration of an embedding we want to maintain the following invariants:
% \vspace{-0.1cm}
\begin{enumerate}[nosep] 
    \item The embedding of any connected component is quasi-correct.
    \item For each tree in the cycle-tree decomposition its embedding respects Septum invariant~\ref{invariant:septum}.
    \item There are no maximal cycles of length $4$.
    \item Each cycle frame is completed with all ``vertical'' edges even if they are not yet revealed.
    \item There are no conflicts with cycle nodes, i.e., each cycle node is the only node mapped to its image in the embedding to $Ladder_N$.
\end{enumerate}

\vspace{-0.1cm}
For each newly revealed edge there are two cases: either it connects two nodes from one connected component or not. We are going to discuss both of them.

\vspace{-0.2cm}
\subsubsection{Edge in one component}
The pseudocode appears in Appendix Algorithm~\ref{alg:one:component:edge}.
If the new edge is already known or it forms a maximal cycle of length four, we simply ignore it. Otherwise, it forms a cycle of length at least six, since two connected nodes are already in one component. 
We then perform the following steps:
\begin{enumerate}
    \item Get the completed frame of a (possibly) new cycle.
    \item Logically ``extract'' it from the component and embed maintaining the orientation (not twisting the core that was already embedded in some way).
    \item Attach two components appeared after an extraction back into the graph, maintaining their relative order.
\end{enumerate}

\vspace{-0.5cm}
\subsubsection{Edge between two components} The pseudocode appears in Appendix Algorithm~\ref{alg:two:components:edge}.
In order to obtain an amortization in the cost, we always ``move'' the smaller component to the bigger one. Thus, the main question here is how to glue a component to the existing embedding of another component. 
The idea is to consider several cases of where the smaller component will be connected to the bigger one. There are three possibilities:
\begin{enumerate}
    \item \textit{It connects to a cycle node.} In this case, there are again two possibilities. Either it ``points away'' from the bigger component meaning that the cycle to which we connect is the one of the ends in the cycle-tree decomposition of the bigger component. Here, we just simply embed it to the end of the cycle-tree decomposition while possibly rotating a cycle at the end. 
    \begin{center}
        \includegraphics[scale=0.25]{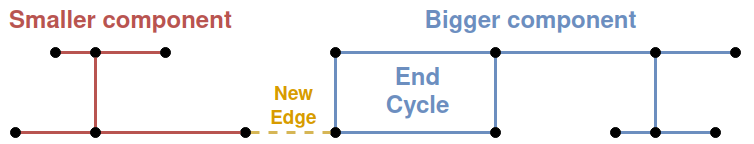}
    \end{center}
    
    Or, the smaller component should be placed somewhere between two cycles in the cycle-tree decomposition. Here, it can be shown that this small graph should be a line-graph, and we can simply add it as a whisker, forming a larger frame. 
    \begin{center}
        \includegraphics[scale=0.25]{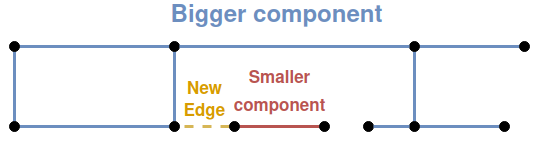}
    \end{center}
    \item \textit{It connects to a trunk core node of a tree in the cycle-tree decomposition.} It can be shown that in this case the smaller component again must be a line-graph. Thus, our only goal is to orient it and possibly two of its inner simple-graphs neighbours to maintain the Septum invariant~\ref{invariant:septum} for the corresponding tree from the decomposition. 
    \begin{center}
        \includegraphics[scale=0.25]{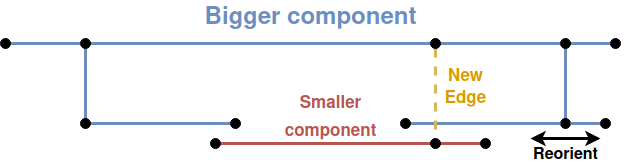}
    \end{center}
    \item \textit{It connects to an exit graph node of an end tree of the cycle-tree decomposition.} %In this case, the corresponding tree from the decomposition of a bigger component is an outer tree in the cycle-tree decomposition.
    In this case, we straightforwardly apply a static embedding algorithm of this tree and the smaller component from scratch. Please, note that only the exit graphs of the end tree will be moved since the trunk core and its inner graphs will remain.
\end{enumerate}

% \vspace{-0.8cm}
% \section{Ladder: cost of the algorithm}
\subsection{Complexity of the online embedding}
\label{sec:cost}

% \vspace{-0.2cm}

Now, we calculate the cost of our online algorithm (a more detailed discussion on the cost of the algorithm appears at Appendix~\ref{app:cost}): how many swaps we should do and how much we should pay for the routing requests. Recall that we first apply the reconfiguration and, then, the routing request.

We start with considering the routing requests. Their cost is $O(1)$ since they lie pretty close on the target line network, i.e., by no more than $12$ nodes apart. This bound holds because the nodes are quasi-correctly embedded on $Ladder_N$, two adjacent nodes at $G$ are located not more than four levels apart (in the worst case, when we remove an edge of a cycle with length four) where each level of the quasi-correct embedding has at most three images of nodes of $G$. Thus, on the target line graph, if we enumerate level by level, the difference between any two adjacent nodes of $G$ is at most $12$.

Then, we consider the reconfiguration. We count the total cost of each case of the online algorithm before all the edges are revealed.

In the first case, we add an edge in one component. By that, either a new frame is created or some frame was enlarged. In both cases, only the nodes, that appear on some frame for the first time, are moved. Since, a node can be moved only once to be mapped on a frame and it is swapped at most $N=O(n)$ times to move to any position, the total cost of this type of reconfiguration is at most $O(n^2)$. Also, there are several adjustments that could be done: 1)~the ``old'' frame can rotate by one node, and 2)~possibly, we should flip the first inner-graphs of two components connected to the frame. In the first modification, each node at the frame can only be ``rotated'' once, thus, paying $O(n)$ cost in total. In the second modification, inner-graph can change orientation at most once in order to satisfy the Septum invariant  (Invariant~\ref{invariant:septum}), thus, paying $O(n^2)$ cost in total~--- each node can move by at most $N=O(n)$. \va{Anton, please, check}

In the second case, we add an edge in between two components. At first, we calculate the time spent on the move of the small component to the bigger one: each node is moved at most $O(\log n)$ times since the size of the component always grows at least two times, the number of swaps of a vertex is at most $N=O(n)$ to move to any place, thus, the total cost is $O(n^2 \log n)$. Secondly, there are two more modification types: 1)~a rotation of a cycle, and 2)~some simple-graphs can be reoriented. The cycle can be rotated only once, thus, we should pay at most $O(n)$ there. At the same time, each simple-graph can be reoriented at most once to satisfy the Septum invariant (Invariant~\ref{invariant:septum}), thus, the total cost is $O(n^2)$ for that type of a reconfiguration.

To summarize, the total cost of requests $\sigma$ is $O(n^2 \log n)$ for the whole reconfiguration plus $O(|\sigma|)$ per requests. This matches the lower bound that was obtained for the line demand graph.
The same result holds for any demand graph that is the subgraph of the ladder of size $n$.

\vspace{-0.2cm}
\begin{theorem}
The online algorithm for embedding the ladder demand graph of size $n$ on the line graph has total cost $O(n^2 \log n + |\sigma|)$ for a sequence of communication requests $\sigma$.
\end{theorem}
\vspace{-0.2cm}
% \begin{remark2}
% The same result holds for any demand graph that is the subgraph of the ladder of size $n$.
% \end{remark2}

\vspace{-0.5cm}
\section{Conclusion}
\label{sec:conclusion}
We presented methods for statically or dynamically re-embedding a ladder demand graph (or a subgraph of it) on a line, both in the offline and online case. As side results, we also presented how to embed a cycle demand graph and a meta-algorithm for a general demand graph.
Our algorithms for the cycle and the ladder cases match the lower bounds.
Our work is a first step towards a tight bound on dynamically re-embedding more generic demand graphs, such as arbitrary grids.

% In this paper, we presented three results: 1)~an online algorithm with its cost for a cycle demand graph; 2)~an online algorithm with its cost for a $Ladder_n$ demand graph; and, finally, 3)~the upper bound cost on any algorithm for an arbitrary demand graph. In the first two cases, we presented algorithms that match the lower bound. We think this is the first important step towards the tight bound for more generic graphs such as arbitrary grids. %that we are going to research in the future work.

%\clearpage
\bibliographystyle{splncs04}
\bibliography{references}

\begin{thebibliography}{10}
\providecommand{\url}[1]{\texttt{#1}}
\providecommand{\urlprefix}{URL }
\providecommand{\doi}[1]{https://doi.org/#1}

\bibitem{avin2022deterministic}
Avin, C., Bienkowski, M., Salem, I., Sama, R., Schmid, S., Schmidt, P.:
  Deterministic self-adjusting tree networks using rotor walks. In: 2022 IEEE
  42nd International Conference on Distributed Computing Systems (ICDCS). pp.
  67--77. IEEE (2022)

\bibitem{avin2019self}
Avin, C., van Duijn, I., Schmid, S.: Self-adjusting linear networks. In:
  International Symposium on Stabilizing, Safety, and Security of Distributed
  Systems. pp. 368--382. Springer (2019)

\bibitem{avin2020complexity}
Avin, C., Ghobadi, M., Griner, C., Schmid, S.: On the complexity of traffic
  traces and implications. Proceedings of the ACM on Measurement and Analysis
  of Computing Systems  \textbf{4}(1),  1--29 (2020)

\bibitem{avin2016online}
Avin, C., Loukas, A., Pacut, M., Schmid, S.: Online balanced repartitioning.
  In: International Symposium on Distributed Computing. pp. 243--256. Springer
  (2016)

\bibitem{avin2022demand}
Avin, C., Mondal, K., Schmid, S.: Demand-aware network design with minimal
  congestion and route lengths. IEEE/ACM Transactions on Networking  (2022)

\bibitem{avin2022push}
Avin, C., Mondal, K., Schmid, S.: Push-down trees: optimal self-adjusting
  complete trees. IEEE/ACM Transactions on Networking  \textbf{30}(6),
  2419--2432 (2022)

\bibitem{avin2019toward}
Avin, C., Schmid, S.: Toward demand-aware networking: a theory for
  self-adjusting networks. ACM SIGCOMM Computer Communication Review
  \textbf{48}(5),  31--40 (2019)

\bibitem{batista2007self}
Batista, D.M., da~Fonseca, N.L.S., Granelli, F., Kliazovich, D.: Self-adjusting
  grid networks. In: 2007 IEEE international conference on communications. pp.
  344--349. IEEE (2007)

\bibitem{diaz2002survey}
D{\'\i}az, J., Petit, J., Serna, M.: A survey of graph layout problems. ACM
  Computing Surveys (CSUR)  \textbf{34}(3),  313--356 (2002)

\bibitem{hansen1989approximation}
Hansen, M.D.: Approximation algorithms for geometric embeddings in the plane
  with applications to parallel processing problems. In: 30th Annual Symposium
  on Foundations of Computer Science. pp. 604--609. IEEE Computer Society
  (1989)

\bibitem{olver2018itinerant}
Olver, N., Pruhs, K., Schewior, K., Sitters, R., Stougie, L.: The itinerant
  list update problem. In: International Workshop on Approximation and Online
  Algorithms. pp. 310--326. Springer (2018)

\bibitem{DBLP:journals/ton/SchmidASBHL16}
Schmid, S., Avin, C., Scheideler, C., Borokhovich, M., Haeupler, B., Lotker,
  Z.: Splaynet: Towards locally self-adjusting networks. {IEEE/ACM} Trans.
  Netw.  \textbf{24}(3),  1421--1433 (2016)

\bibitem{sleator1985amortized}
Sleator, D.D., Tarjan, R.E.: Amortized efficiency of list update and paging
  rules. Communications of the ACM  \textbf{28}(2),  202--208 (1985)

\end{thebibliography}

\clearpage
\appendix

\section{The algorithm for the Cycle}
\label{app:cycle}
We start with the most simple generalization result~--- when the demand graph is the cycle on $n$ vertices. 

{
\def\thetheorem{\ref{thm:cycle}}
\begin{theorem}
Suppose the demand graph is $C_n$. There is an algorithm such that the total cost spent on the migrations is $O(n^2 \log n)$ and each request is performed in $O(1)$.
In particular, if the number of requests is $\Omega(n^2 \log n)$ each request has $O(1)$ amortized cost.
\end{theorem}
\addtocounter{theorem}{-1}
}

\begin{proof}

The idea of the algorithm is to act as in the algorithm described in~\cite{avin2019self} for the list demand graph until revealed edges do not form a cycle. Once they do we perform a total reconfiguration enumerating nodes of $C_n$ with
\begin{align*}
    \begin{cases}
        i \rightarrow 2i - 1, \text{ if } i \leq \lceil \frac{n}{2} \rceil\\
        i \rightarrow 2(n - i + 1), \text{ otherwise}
    \end{cases}
\end{align*}
so, for each pair of adjacent nodes the difference of their numbers is at most 2. The enumeration can be seen on Figure~\ref{fig:cycle}.

% \begin{figure}
% \begin{center}
%     \includegraphics[scale=0.4]{images/cycle-numeration.png}
% \end{center}
% \caption{Cycle enumeration with Bandwidth $2$}
% \label{fig:cycle}
% \end{figure}

% \begin{wrapfigure}{r}{0.4\textwidth}
%     \includegraphics[scale=0.6]{images/cycle-numeration.png}
% \caption{Cycle enumeration with Bandwidth $2$}
% \label{fig:cycle}
% \end{wrapfigure}

More formally. Let $G_{i} = (V, E_i)$~--- the demand graph after $i$ requests, and $G_0 = (V, \emptyset)$.
We want to maintain the invariant that each $G_i$ is embedded in a way that all adjacent 
nodes are at a distance of at most $2$. Moreover, if there is a line subgraph of $G_i$ then it is embedded as a line, i.e., the embedding preserves edges. We present an algorithm that maintains this invariant by induction.
\InsertBoxR{0}{\begin{minipage}{0.45\linewidth}\centering
\begin{center}
\captionsetup{type=figure}
\includegraphics[scale=0.4]{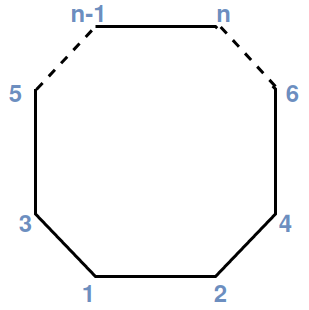}    
\end{center}
\captionof{figure}{\small Cycle enumeration with Bandwidth $2$}
\label{fig:cycle}
% \end{figure}
\end{minipage}%
}[6]

This invariant holds for $G_0$.
We assume that the invariant holds for $G_{i-1}$ and a new request $\sigma_i$ arrives.
If $\sigma_i$ is already present in $E(G_{i-1})$ then the invariant holds, we do not reconfigure, and pay at most $2$.

If now $G_i$ is a cycle we perform a total reconfiguration with the enumeration with bandwidth $2$ described above. For that we pay $O(n^2)$ that is less than $O(n^2 \log n)$, and, thus, our complexity lies inside our bounds. Note that once $G_i$ becomes a cycle we need no further reconfigurations since all the edges are known and the invariant is maintained. 

The last case is when $\sigma_i$ is a new edge and $G_i$ still consists of several connected components. We use the algorithm presented in~\cite{avin2019self}. $\sigma_i$ connects two different connected components, say $L_1$ and $L_2$ forming a new list subgraph $L$. Suppose that $|V(L_1)| \leq |V(L_2)|$. Our strategy would be to ``drag'' $L_1$ towards $L_2$ that is if $L_1 = \{u_1, \ldots , u_l\}$, $V(L_2) = \{v_1, \ldots , v_k\}$, $\sigma_i = (v_k, u_1)$. By the invariant $L_2$ is embedded with $v_p$ at $q + p$ for some $q$ and we want nodes of $L_1$ to be embedded with  $u_p \rightarrow q + k - 1 + p$.
So, we bring each node of $L_1$ to its position performing required number of swaps. 
Note, that this reconfiguration brings the embedding that supports the invariant.
%Note that for each pair of nodes not in $L_1$ their relative order did not change and furthermore if the two nodes were adjacent in a list subgraph of $G_{i-1}$ they are still adjacent in the resulting embedding, so, the invariant is maintained.
Now, we analyze the cost of the algorithm processing the requests.

Due to the invariant each request is served with a cost of at most $2$. As for the reconfiguration cost: a node can move a distance $\Theta(n)$ during processing one request and it moves no more than $O(\log n)$ times: we either form a cycle or merge two components. Thus, the total cost does not exceed $O(n^2 \log n)$.
%We make a cycle only once, so that is just $+1$ to the total amount. And if a node moves during the merge of two connectivity components: it is a part of a smaller component and, thus, by a new edge its component size is at least doubled meaning that each node cannot move more than $\lfloor \log n \rfloor$ times.
%By that, each node performs in modifications no more than $O(n \log n)$ cost and, thus, the whole reconfiguration cost is bounded with $O(n^2 \log n)$.
\end{proof}

\section{Full algorithm for the ladder}

\subsection{Static quasi-embedding}
\label{app:static}
\subsubsection{Tree embedding}
Assume, we are given a tree $T$ that can be embedded into $Ladder_n$. Furthermore, there are two marked nodes in the tree: one is marked \textit{right} and the \textit{left}. It is known that there is a correct embedding of $T$ with \textit{right} being the right-most node, meaning no node is embedded higher or to the same level, and \textit{left} being the left-most node.

We now describe how to obtain a quasi-correct embedding of $T$ with \textit{right} being the right-most node and \textit{left} being the left-most one and $left$ mapped to $leftImage$. Moreover, this embedding obeys the following invariant:

\begin{invariant}[Septum invariant]
For each inner simple-graph its foot and its head are embedded to the same level and no other node is embedded to that level.
\end{invariant}

We embed the $left-right$ path strictly vertically and then we orient line-graphs connected to it in the way that they do not violate septum invariant.

See Algorithm~\ref{alg:LR:tree:embedding}.
\begin{algorithm}
\caption{Left-Right tree embedding}
\begin{algorithmic}
\Procedure{RightLeftTreeEmbedding}{$T$, $left$, $right$, $leftImage$}
    \State $P \gets$ path from \textit{left} to \textit{right}
    \State $leftSide \gets side(leftImage)$
    \State $leftLevel \gets level\langle leftImage \rangle$
    \State Embed $P[i] \rightarrow level(leftLevel - 1 + i)[leftSide]$
    \State $L \gets$ line-graphs connected to $P$
    \State $Septa \gets \{level\langle foot(l) \rangle \mid l \in L\} \cup \{level\langle left \rangle, level\langle right \rangle\}$
    \For{$l \in L$}
        \State $i \gets level\langle foot(l) \rangle$
        \State Embed $head(l) \rightarrow level(i)[other(leftSide)]$
        \State Orient $l$ ensuring no nodes of $l \setminus \{head(l)\}$ are embedded to any of levels from $Septa$
    \EndFor
\EndProcedure
\end{algorithmic}
\label{alg:LR:tree:embedding}
\end{algorithm}

We now proceed with an embedding of a tree $T$ where $right$, $left$ and $leftImage$ might or might not be given. In the case the variable is not given we denote its value with $\None$.

The idea here would be to first embed the trunk core and its inner line-graphs using Left-Right tree embedding and then to embed exit-graphs strictly vertically "away" from the trunk core. That means that the hands of exit-graphs that are connected to to the right of the trunk core are embedded increasingly and hands of those exit-graphs which are connected to the left of the trunk core are embedded decreasingly. 

If a tree does not have a trunk core, that means that its structure is rather simple (in particular it has no more than two nodes of degree three), so we do not care about conflicts.

See Algorithm~\ref{alg:tree:embedding}.

\begin{breakablealgorithm}
\caption{Tree quasi-correct embedding}
\begin{algorithmic}
\Procedure{TreeQuasiCorrectEmbedding}{$T$, $left$, $right$, $leftImage$}
    \If{$leftImage = \None$}
        \State $leftImage \gets level(1)[1]$
    \EndIf
    
    \If{$(left \neq \None) \wedge (right \neq \None)$}
        \State LeftRightTreeEmbedding($T$, $left$, $right$, $leftImage$)
        \State \Return
    \ElsIf{$(left \neq \None) \wedge (right = \None)$}
        \If{$T$ has a trunk core}
            \State $u, v \gets$ ends of a trunk core
            \If{$u$ between $v$ and $left$}
                \State $trunkRight \gets v$
            \Else
                \State $trunkRight \gets u$
            \EndIf
            \State $et \gets$ exit-graphs connected to $trunkRight$
            \State $S' \gets S \setminus et$
            \State LeftRightTreeEmbedding($S'$, $left$, $trunkRight$, $leftImage$)
            \For{$e \in et$}
                \State $lvl \gets level\langle trunkRight \rangle$
                \State $side \gets side(trunkRight)$
                \State Embed $head(e) \rightarrow level(lvl + 1)[side]$
                \For{$h \in hands(e)$}
                    \For{$i \in [length(h)]$}
                        \State Embed $h[j] \rightarrow level(lvl + 1 + i)[side]$
                    \EndFor
                \EndFor
            \EndFor
        \Else
            \State $right \gets$ arbitrary node of degree 1
            \State LeftRightTreeEmbedding($T$, $left$, $right$, $leftImage$)
        \EndIf
    \ElsIf{$(left = \None) \wedge (right \neq \None)$}
        \If{$T$ has a trunk core}
            \State $u, v \gets$ ends of a trunk core
            \If{$u$ between $v$ and $right$}
                \State $trunkLeft \gets v$
            \Else
                \State $trunkLeft \gets u$
            \EndIf
            \State $eb \gets$ exit-graphs connected to $trunkLeft$
            \State $S' \gets S \setminus eb$
            \State $leftImageLevel \gets level\langle leftImage\rangle$
            \State $leftImageSide \gets side(leftImage)$
            \State $vShift \gets \max\limits_{e \in eb}\max\limits_{h \in hands(e)}length(h)$
            \State $trunkLeftImage \gets level(leftImageLevel + vShift)[leftImageSide]$
            \State LeftRightTreeEmbedding($S'$, $left$, $right$, $trunkLeftImage$)
            \For{$e \in eb$}
                \State $lvl \gets level\langle trunkLeft \rangle$
                \State $side \gets side(trunkLeft)$
                \State Embed $head(e) \rightarrow level(lvl - 1)[side]$
                \For{$h \in hands(e)$}
                    \For{$i \in [length(h)]$}
                        \State Embed $h[j] \rightarrow level(lvl - 1 - i)[side]$
                    \EndFor
                \EndFor
            \EndFor
        \Else
            \State $left \gets$ arbitrary node of degree 1
            \State LeftRightTreeEmbedding($T$, $left$, $right$, $leftImage$)
        \EndIf
    \Else
        \If{$T$ has a trunk core}
            \State $trunkRight,\ trunkLeft \gets$ ends of a trunk core
            \State $et \gets$ exit-graphs connected to $trunkRight$
            \State $eb \gets$ exit-graphs connected to $trunkLeft$
            \State $S' \gets S \setminus et \setminus eb$
            \State $leftImageLevel \gets level\langle leftImage\rangle$
            \State $leftImageSide \gets side(leftImage)$
            \State $vShift \gets \max\limits_{e \in eb}\max\limits_{h \in hands(e)}length(h)$
            \State $trunkLeftImage \gets level(leftImageLevel + vShift)[leftImageSide]$
            \State LeftRightTreeEmbedding($S'$, $left$, $right$, $trunkLeftImage$)
            \For{$e \in et$}
                \State $lvl \gets level\langle trunkRight \rangle$
                \State $side \gets side(trunkRight)$
                \State Embed $head(e) \rightarrow level(lvl + 1)[side]$
                \For{$h \in hands(e)$}
                    \For{$i \in [length(h)]$}
                        \State Embed $h[j] \rightarrow level(lvl + 1 + i)[side]$
                    \EndFor
                \EndFor
            \EndFor
            \For{$e \in eb$}
                \State $lvl \gets level\langle trunkLeft \rangle$
                \State $side \gets side(trunkLeft)$
                \State Embed $head(e) \rightarrow level(lvl - 1)[side]$
                \For{$h \in hands(e)$}
                    \For{$i \in [length(h)]$}
                        \State Embed $h[j] \rightarrow level(lvl - 1 - i)[side]$
                    \EndFor
                \EndFor
            \EndFor
        \Else
            \State $right, left \gets$ furthest nodes of degree 1
            \State $P \gets$ path from $left$ to $right$
            \State Embed $P$ strictly monotone placing $left$ to $leftImage$
            \If{there is a line-graph left}
                \State $l \gets$ left line-graph
                \State $(u, v) \gets (u, v) \in E(T)$ s.t. $u \in P, v \in l$
                \State Embed $v$ to the same level, opposite side to $u$.
                \State Embed $l$ to the opposite side to the side of $P$ preserving connectivity and correctness
            \EndIf
        \EndIf
    \EndIf
\EndProcedure
\end{algorithmic}
\label{alg:tree:embedding}
\end{breakablealgorithm}

\subsubsection{Cycle embedding}
Given a cycle $C$ of length $\geq 6$ and nodes $left, right \in V(C)$ we construct a correct embedding of $C$ into $Ladder_\infty$  with $level\langle left \rangle \leq level\langle u \rangle \leq level\langle right \rangle\ \forall u \in V(C)$ and $left$ placed to $leftImage$.

For convenience we assume that for every node $v$ there is a local consecutive numeration starting at $v$. The number of node $u \in V(C)$ in this numeration is referenced with $number_v(u)$. The node with number $i$ in local numeration of $v$ is referenced with $C_v[i]$.

We first check if it is possible to satisfy the given constraints of placing the $left$ node to left and a $right$ node to the right. If it is indeed possible, we place $left$ to desired place and then choose an orientation (clockwise or counterclockwise) following which we would embed the rest of the nodes, keeping in mind that $right$ must stay on the highest level. See Algorithm~\ref{alg:LR:cycle:embedding}.

\begin{algorithm}
\caption{Left-Right cycle embedding}
\begin{algorithmic}
\Procedure{LeftRightCycleEmbedding}{$C$, $left$, $right$, $leftImage$}
    \State $h \gets \frac{length(C)}{2}$
    \State \textbf{Ensure:} $number_{left}(right) \in \{h, h + 1, h + 2\}$
    
    \State $leftLevel \gets level\langle leftImage \rangle$
    \State $leftSide \gets side(leftImage)$
    \If{$(number_{left}(right) = h) \vee(number_{left}(right) = h + 1)$}
        \For{$i \in [h]$}
            \State Embed $C_{left}[i]\rightarrow level(leftLevel + i - 1)[leftSide]$
        \EndFor
        \For{$i \in [h]$}
            \State Embed $C_{left}[h + i] \rightarrow level(leftLevel + h - i)[other(leftSide)]$
        \EndFor
    \Else
        \State Embed $left \rightarrow leftImage$
        \For{$i \in [h]$}
            \State Embed $C_{left}[i + 1]\rightarrow level(leftLevel + i - 1)[other(leftSide)]$
        \EndFor
        \For{$i \in [h - 1]$}
            \State Embed $C_{left}[h + 1 + i]\rightarrow level(leftLevel + h - i)[leftSide]$
        \EndFor
    \EndIf
\EndProcedure
\end{algorithmic}
\label{alg:LR:cycle:embedding}
\end{algorithm}

Now, suppose that not all information, such as $right$, $left$, and $LeftImage$, is provided.
We will reduce this problem to the case when the missing variables are known. Though the subtlety might occur due to the fact that there are inner edges in the cycle. In this case we choose missing $left/right$ more precisely in order to embed inner edge vertically. See Algorithm~\ref{alg:cycle:embedding}.

\begin{algorithm}
\caption{Cycle embedding}
\begin{algorithmic}
\Procedure{CycleEmbedding}{$C$, $left$, $right$, $leftImage$}
    \State $h \gets \frac{length(C)}{2}$
    \If{$(left \neq \None) \wedge (right \neq \None)$}
        \State \textbf{Ensure:} $number_{left}(right) \in \{h, h + 1, h + 2\}$
    \EndIf

    \If{$leftImage = \None$}
        \State $leftImage \gets level(1)[1]$
    \EndIf
    
    \If{$(left = \None) \wedge (right = \None)$}
        \State $left \gets$ arbitrary node of $C$
        \If{$C$ has an inner edge}
            \State Choose $right$ out of $\{C_{left}[h], C_{left}[h + 2]\}$ to respect the inner edge
        \Else
            \State Choose $right$ out of $\{C_{left}[h], C_{left}[h + 2]\}$ arbitrary
        \EndIf
    \ElsIf{$(left \neq \None) \wedge (right = \None)$}
        \If{$C$ has an inner edge}
            \State Choose $right$ out of $\{C_{left}[h], C_{left}[h + 2]\}$ to respect the inner edge
        \Else
            \State Choose $right$ out of $\{C_{left}[h], C_{left}[h + 2]\}$ arbitrary
        \EndIf
    \ElsIf{$(left = \None) \wedge (right \neq \None)$}
        \If{$C$ has an inner edge}
            \State Choose $left$ out of $\{C_{left}[h], C_{left}[h + 2]\}$ to respect the inner edge
        \Else
            \State Choose $left$ out of $\{C_{left}[h], C_{left}[h + 2]\}$ arbitrary
        \EndIf
    \EndIf
    
    LeftRightCycleEmbedding($C$, $left$, $right$, $leftImage$)
\EndProcedure
\end{algorithmic}
\label{alg:cycle:embedding}
\end{algorithm}

\subsubsection{Component embedding}
Right now we explain on how to embed onto $Ladder_N$ a connectivity component $S$ that can be embedded onto $Ladder_n$.

We start with an algorithm on how to make a cycle-tree decomposition chain of $S$ assuming no uncompleted frames.
To obtain a cycle-tree decomposition of a graph: 1)~we find a maximal cycle; 2)~we split the graph into two parts by logically removing the cycle; 3)~we proceed recursively on those parts, and, finally, 4)~we combine the results together maintaining the correct order of the chain components. See the Algorithm~\ref{alg:decomposition}.

\begin{algorithm}
\caption{Cycle-Tree decomposition chain}
\begin{algorithmic}
\Function{CycleTreeDecompositionChain}{$S$}
    \Ensure $S$ has no uncompleted frames
    \If{$S$ is empty}
        \State \Return $[]$
    \ElsIf{$S$ is a tree}
        \State \Return $[S]$
    \Else
        \State $C \gets$ arbitrary maximal cycle in $S$
        \State $S_1, S_2 \gets$ connectivity components of $S \setminus C$
        \State $C_1\gets CycleTreeDecompositionChain(S_1)$
        \State $C_2 \gets CycleTreeDecompositionChain(S_2)$
        \If{$C_2$ is empty}
            \State \Return $[C] + C_1$
        \Else
            \If{$\exists u \in V(C_2[0]),\ v \in V(C),\ s.t.\ (u, v) \in E(S)$}
                \State \Return $C_1 + [S] + C_2$
            \Else
                \State \Return $C_2 + [S] + C_1$
            \EndIf
        \EndIf
    \EndIf
\EndFunction
\end{algorithmic}
\label{alg:decomposition}
\end{algorithm}

Now, we describe how to obtain a quasi-correct embedding of $S$. We preprocess $S$: 1)~we remove one edge from cycles of size four; 2)~we complete uncompleted frames with vertical edges. After this preprocessing, we embed parts of $S$ from the cycle-tree decomposition chain one by one in the relevant order using the corresponding algorithm (either for a cycle or for a tree embedding) making sure parts are glued together correctly. 

As before, we have additional variables $left$, $right$ and $leftImage$ which might or might not be given.

\begin{breakablealgorithm}
\caption{Connectivity component quasi-correct embedding}
\begin{algorithmic}
\Procedure{Preprocess}{$S$}
    \State $C \gets$ maximal cycles of length 4 in $S$
    \For{$c \in C$}
        \State remove arbitrary edge of $c$ from $S$
    \EndFor
    \State $F \gets$ cycle frames in $S$
    \For{$f \in F$}
        \State complete $F$
    \EndFor
\EndProcedure

\Procedure{ComponentEmbeddingLeftFixed}{$S$, $left$, $right$, $leftImage$}
    \If{$S$ is a tree}
        \State TreeQuasiCorrectEmbedding($S$, $left$, $right$, $leftImage$)
        \State \Return
    \EndIf
    
    \If{$S$ is a cycle}
        \State CycleEmbedding($S$, $left$, $right$, $leftImage$)
        \State \Return
    \EndIf
    
    \State Preprocess($S$)
    
    \State $C \gets CycleTreeDecompositionChain(S)$
    \If{$left \neq \None$}
        \State Reverse $C$ in the way that $left \in C[1]$
    \EndIf
    \If{$right \neq \None$}
        \State Reverse $C$ in the way that $right \in C[length(C)]$
    \EndIf
    
    \For{$i \in [length(C)]$}
        \If{$i = 1$}
            $u, v \gets (u, v) \in E(S)$, s.t. $(u \in V(C[1])) \wedge (v \in V(C[i + 1])$
            \If{$C[1]$ is a tree}
                \State $cur \gets C[1] \cup (u, v)$
                \State TreeQuasiCorrectEmbedding($cur$, $left$, $v$, $leftImage$)
            \Else
                \State CycleEmbedding($C[1]$, $left$, $u$, $leftImage$)
            \EndIf
        \ElsIf{$i = length(C)$}
            \State $u, v \gets (u, v) \in E(S)$, s.t. $(u \in V(C[i - 1])) \wedge (v \in V(C[i])$
            \State $leftLevel \gets level\langle u \rangle + 1$
            \State $leftSide \gets side(u)$
            \State $localLeftImage \gets level(leftLevel)[leftSide]$
            \If{$C[i]$ is a cycle}
                \State CycleEmbedding($C[i]$, $v$, $right$, $localLeftImage$)
            \Else
                \State TreeQuasiCorrectEmbedding($C[i]$, $v$, $right$, $localLeftImage$)
            \EndIf
        \Else
            \State $u_1, v_1 \gets (u, v) \in E(S)$, s.t. $(u \in V(C[i - 1])) \wedge (v \in V(C[i])$
            \State $leftLevel \gets level\langle u_1 \rangle + 1$
            \State $leftSide \gets side(u_1)$
            \State $localLeftImage \gets level(leftLevel)[leftSide]$
            
            \State $u_2, v_2 \gets (u, v) \in E(S)$, s.t. $(u \in V(C[i])) \wedge (v \in V(C[i + 1])$
            \If{$C[i]$ is a cycle}
                \State CycleEmbedding($C[i]$, $v_1$, $u_2$, $localLeftImage$)
            \Else
                \State $cur \gets C[i] \cup (u_2, v_2)$
                \State TreeQuasiCorrcetEmbedding($C[i]$, $v_1$, $v_2$, $localLeftImage$)
            \EndIf
        \EndIf
    \EndFor
\EndProcedure
\end{algorithmic}
\label{alg:component:embedding}
\end{breakablealgorithm}

We finally notice that having a procedure to embed a component with a fixed $leftImage$ it is easy to obtain a procedure which embeds with $rightImage$ fixed. We simply apply the "$left$" procedure and then flip the result.

\begin{algorithm}
\caption{Component embedding right fixed}
\begin{algorithmic}
\Procedure{ComponentEmbeddingRightFixed}{$S$, $left$, $right$, $rightImage$}
    \State ComponentEmbeddingLeftFixed($S$, $right$, $left$, $rightImage$)
    \State Flip the image of $S$ over horizontal axis maintaining the position of $right$
\EndProcedure
\end{algorithmic}
\end{algorithm}

\subsection{Dynamic algorithm}
\label{app:dynamic}
We describe how one should change the embedding of the graph after the processing of a request in an online scenario. At each moment we have some edges of a $Ladder_n$ already revealed forming connectivity components.
After an edge reveal we should reconfigure the target line graph. For that, instead of line reconfiguration we reconfigure our embedding to $Ladder_N$ that is then embedded to the line line by line and introduce some constant factor. So, we can consider the reconfiguration only of $Ladder_N$ and forget about the target line graph at all.
When doing the reconfiguration of an embedding we want to maintain the following invariants:
\begin{enumerate} 
    \item The embedding of any connectivity component is quasi-correct.
    \item For each tree in the cycle-tree decomposition its embedding respects Septum invariant~\ref{invariant:septum}.
    \item There are no maximal cycles of length $4$.
    \item Each cycle frame is completed with all ``vertical'' edges even if they are not yet revealed.
    \item There are no conflicts with cycle nodes, i.e., two nodes of a cycle do not map to same node of $Ladder_N$.
\end{enumerate}

For each newly revealed edge there are two cases: either it connects two nodes from one connectivity component or not. We are going to discuss both of them.

\subsubsection{Edge in one component}

If the new edge is already known or it forms a maximal cycle of length four, we simply ignore it. Otherwise, it forms a cycle of length at least six, since two connected nodes are already in one component. 

We then perform the following steps:
\begin{enumerate}
    \item Get the completed frame of a (possibly) new cycle.
    \item Logically ``extract'' it from the component and embed maintaining the orientation (not twisting the core that was already embedded in some way).
    \item Attach two components appeared after an extraction back into the graph, maintaining their relative order.
\end{enumerate}

\begin{algorithm}[H]
\caption{Process Edge In One Component}
\begin{algorithmic}
\Procedure{ProcessEdgeOneComponent}{$S$, $(u, v)$}
\If{Edge $(u, v)$ already exists}
    \State \Return
\EndIf
\If{Edge $(u, v)$ forms a maximal cycle of length 4}
    \State \Return
\EndIf

\State $C \gets$ maximal cycle containing $u, v$
\State $F \gets$ completed frame of $C$
\State $S_1, S_2 \gets$ connectivity components of $S \setminus F$
\State $u_1, v_1 \gets u, v:\ u \in V(F), v \in V(S_1), (u, v) \in E(S)$
\State $u_2, v_2 \gets u, v:\ u \in V(F), v \in V(S_2), (u, v) \in E(S)$
\If{$level\langle u_1 \rangle > level\langle u_2 \rangle$}
    \State $Swap(S_1, S_2),\ Swap(u_1, u_2),\ Swap(v_1, v_2)$
\EndIf
\State CycleEmbedding($F$, $u_1$, $u_2$, $\None$)
\State ComponentEmbeddingTopFixed($S_1$, $\None$, $u_1$, $image(u_1)$)
\State ComponentEmbeddingBotFixed($S_2$, $u_2$, $\None$, $image(u_2)$)
\EndProcedure
\end{algorithmic}
\label{alg:one:component:edge}
\end{algorithm}

\subsubsection{Edge between two components}

In order to obtain an amortization in the cost, we always ``move'' the smaller component to the bigger one. Thus, the main question here is how to glue a component to the existing embedding of another component. 

The idea is to consider several cases of where the smaller component will be connected to the bigger one. There are three possibilities:
\begin{enumerate}
    \item \textit{It connects to a cycle node.} In this case there are again two possibilities. Either it ``points away'' from the bigger component meaning that the cycle to which we connect is the one of the ends in the cycle-tree decomposition of the bigger component. Here, we just simply embed it to the end of the cycle-tree decomposition while possibly rotating a cycle at the end. 
    \begin{center}
        \includegraphics[scale=0.4]{images/end-cycle.png}
    \end{center}
    
    Or, the smaller component should be placed somewhere between two cycles in the cycle-tree decomposition. Here, it can be shown that this small graph should be a line-graph, and we can simply add it as a whisker, forming a larger frame. 
    \begin{center}
        \includegraphics[scale=0.4]{images/between-cycles.png}
    \end{center}
    \item \textit{It connects to a trunk core node of a tree in the cycle-tree decomposition.} It can be shown that in this case the smaller component again must be a line-graph. Thus, our only goal is to orient it and possibly two of its inner simple-graphs neighbours to maintain the Septum invariant~\ref{invariant:septum} for the corresponding tree from the decomposition. 
    \begin{center}
        \includegraphics[scale=0.4]{images/to-trunk-core.png}
    \end{center}
    \item \textit{It connects to an exit graph node of an end tree of the cycle-tree decomposition.} %In this case, the corresponding tree from the decomposition of a bigger component is an outer tree in the cycle-tree decomposition.
    In this case, we straightforwardly apply a static embedding algorithm of this tree and the smaller component from scratch. Please, note that only the exit graphs of the end tree will be moved since the trunk core and its inner graphs will remain.
\end{enumerate}

\begin{breakablealgorithm}
\caption{Process edge between two components}
\begin{algorithmic}
\Procedure{AddInnerWhisker}{$S$, $C$, $W$, $(u, v)$}
    \State $S \gets S \cup W \cup (u, v)$
    \State $F \gets$ completed frame of $C$
    \State $S_1, S_2 \gets$ connectivity components of $S \setminus F$
    \If{$S_1$ is embedded above $S_2$}
        \State $Swap(S_1, S_2)$
    \EndIf
    \State $s_1, t_1 \gets s, t:\ s \in V(F), t \in V(S_1), (s, t) \in E(S)$
    \State $s_2, t_2 \gets s, t:\ s \in V(F), t \in V(S_2), (s, t) \in E(S)$
    
    \State CycleEmbedding($F$, $s_1$, $s_2$, $\None$)
    \State ComponentEmbeddingTopFixed($S_1 \cup (s_1, t_1)$, $\None$, $s_1$, $image(s_1)$)
    \State ComponentEmbeddingBotFixed($S_2 \cup (s_2, t_2)$, $s_2$, $\None$, $image(s_2)$)
\EndProcedure

\Procedure{ProcessEdgeTwoComponents}{$S_1$, $S_2$, $(u, v)$}
    \State \textbf{Ensure:} $u \in V(S_1), v \in V(S_2)$
    \If{$V(S_1) < V(S_2)$}
        \State $Swap(S_1, S_2),\ Swap(u, v)$
    \EndIf
    
    \State $DC_1 \gets CycleTreeDecomposition(S_1)$
    \State Reverse $DC_1$ in a way $DC[i]$ is embeded under $DC[i + 1]\ \forall i$
    \State $A, i \gets DC_1[i], i:\ u \in V(DC_1[i])$
    
    \If{$A$ is a cycle}
        \If{$length(DC_1) = 1$}
            \If{$A$ has an inner edge}
                \If{$u$ is a top node}
                    \State $bot \gets$ arbitrary bottom node of $A$
                    \State ComponentEmbeddingBotFixed($S_1 \cup S_2 \cup (u, v)$, $bot$, $\None$, $\None$)
                \Else
                    \State $top \gets$ arbitrary top node of $A$
                    \State ComponentEmbeddingTopFixed($S_1 \cup S_2 \cup (u, v)$, $\None$, $top$, $\None$)
                \EndIf
            \Else
                \State ComponentEmbeddingBotFixed($S_1 \cup S_2 \cup (u, v)$, $\None$, $\None$, $\None$)
            \EndIf
        \ElsIf{$i = 1$}
            \State $p, q \gets p, q:\ p \in V(DC_1[i]), q \in V(DC_1[i + 1]), (p, q) \in E(S_1)$
            \If{$(u, p) \in E(S_1)$}
                AddInnerWhisker($S_1$, $A$, $S_2$, $(u, v)$)
            \Else
                \If{$u$ is not a bottom node of $A$}
                    \State Flip $A$ over diagonal containing $p$
                \EndIf
                \State ComponnetEmbeddingTopFixed($S_2$, $\None$, $u$, $image(u)$)
            \EndIf
        \ElsIf{$i = length(DC_1)$}
            \State $p, q \gets p, q:\ p \in V(DC_1[i]), q \in V(DC_1[i - 1]), (p, q) \in E(S_1)$
            \If{$(u, p) \in E(S_1)$}
                AddInnerWhisker($S_1$, $A$, $S_2$, $(u, v)$)
            \Else
                \If{$u$ is not a top node of $A$}
                    \State Flip $A$ over diagonal containing $p$
                \EndIf
                \State ComponnetEmbeddingBotFixed($S_2$, $u$, $\None$, $image(u)$)
            \EndIf
        \Else
            AddInnerWhisker($S_1$, $A$, $S_2$, $(u, v)$)
        \EndIf
    \EndIf
    
    \If{$A$ is a tree}
        \If{$u \in$ extended trunk core of $A$}
            \State Embed $v \rightarrow opposite(u)$
            \State $l_1, l_2 \gets$ $u$ neighbouring inner simple-graphs
            \State Orient $S_2, l_1, l_2$ to maintain Septum invariant in $A$
        \ElsIf{$u \in$ inner simple-graph}
            \State $S_1 \gets S_1 \cup S_2 \cup (u, v)$
            \State $l \gets$ inner simple graph containing $u$
            \State Orient $l$ to maintain Septum invariant in $A$
        \ElsIf{$u \in$ exit-graph}
            \If{$i = 1$}
                \If{$length(DC_1) = 1$}
                    \State $p \gets$ arbitrary highest node of $S_1$
                    \State $q \gets$ additional temporary node
                \Else
                    \State $p, q \gets p, q:\ p \in V(DC_1[i]), q \in V(DC_1[i + 1]), (p, q) \in E(S_1)$
                \EndIf
                \State ComponentEmbeddingTopFixed($A \cup S_2 \cup (p, q)$, $\None$, $q$, $image(q)$)
            \ElsIf{$i = length(DC_1)$}
                \State $p, q \gets p, q:\ p \in V(DC_1[i]), q \in V(DC_1[i - 1]), (p, q) \in E(S_1)$
                \State ComponentEmbeddingBotFixed($A \cup S_2 \cup (p, q)$, $q$, $\None$, $image(q)$)
            \EndIf
        \EndIf
    \EndIf

\EndProcedure
\end{algorithmic}
\label{alg:two:components:edge}
\end{breakablealgorithm}

\section{Proofs and Analysis}
\label{app:proofs}
\subsection{Strategy}\label{strategy}
At the very beginning there are no requests and we don't know any request-edges. Requests come one at a time, possibly, revealing new edges. Known edges form connectivity components which are all subgraphs of the request graph. Our strategy would be to maintain such enumeration $\sigma$ of vertices that for each connectivity component $S$
\begin{align}
\max\limits_{(u, v) \in E(S)} |\sigma(u) - \sigma(v)| \leq 12
\end{align}
We call this property of an enumeration the \textit{proximity property}. 

So, if we receive the request which was already known, we do nothing since the property persists. But if the new edge comes, we might perform a re-enumeration $\sigma$ on the vertices to maintain the property.

\subsection{Bandwidth of subgraphs}
\label{ap:bandwidth}

%% moved to main text 
%
% {
% \def\thelemma{\ref{lem:grid bandwidth}}
% \begin{lemma}
% $\bandwidth(Ladder_n) = 2$.
% \end{lemma}
% \addtocounter{lemma}{-1}
% }
% %

% \begin{proof}
% \label{lem:grid bandwidth}
% %[Lemma \ref{lem:grid bandwidth}]
% %\label{proof:lem:grid bandwidth}
% The bandwidth is greater than 1, because there are nodes of degree three. The bandwidth of 2 can be achieved via the ``level-by-level'' enumeration as shown in the figure below
% \begin{figure}
%     \centering
%     \includegraphics[scale=0.4]{images/ladder-numeration.png}
%     \caption{Optimal ladder numeration}
%     \label{fig:ladder_numeration}
% \end{figure}
% \end{proof}

\begin{definition}
Consider two connected graphs $S$ and $G$. The correct embedding of $S$ into $G$ is a mapping $\varphi: V(S) \rightarrow V(G)$ such that:
\begin{itemize}
    \item $\varphi$ is injective
    \item $(u, v) \in E(S) \rightarrow (\varphi(u), \varphi(v)) \in E(G)$
\end{itemize}

If $\varphi$ is not injective, i.e. there are nodes $u, v$, s.t. $\varphi(u) = \varphi(v)$, we say that there is a conflict between $u$ and $v$.
\end{definition}

{
\def\thelemma{\ref{lem:embedding enumeration}}
\begin{lemma}
For each subgraph $S$ of a graph $G$, $\bandwidth(S) \leq \bandwidth(G)$.
\end{lemma}
\addtocounter{lemma}{-1}
}
\begin{proof}
% \label{lem:embedding enumeration}
%\label{proof:lem:embedding enumeration}
%[Proof of Lemma \ref{lem:embedding enumeration}]
Let $\varphi$  be a correct embedding of $S$ into $G$. And let $\sigma$ be the enumeration on $G$ with which the $bandwidth(G)$ is achieved.

Let $U$ be the finite set of unique natural numbers. For $v \in U$ we define $ord_U(v) = |\{u \mid u \in U, u \leq v\}|$.

We now define the enumeration $\sigma_S$ of $S$ as follows: 
\begin{align*}
    U = \{\sigma(\varphi(v)) \mid v \in S\}\\
    \sigma_S(v) = ord_U(\sigma(\varphi(v)))
\end{align*}

We state that $\max\limits_{(u,v) \in E(S)}|\sigma_S(u) - \sigma_S(v)| \leq bandwidth(G)$. This follows from two facts:
\begin{itemize}
    \item For every $(u, v) \in E(S)$ 
    \begin{align*}
    |\sigma(\varphi(u)) - \sigma(\varphi(v))| \leq bandwidth(G) 
    \end{align*}
    since $(\varphi(u), \varphi(v)) \in E(G)$
    \item If U is a set of unique natural numbers than for every $u, v \in U$ 
    \begin{align*}
    |ord_U(u) - ord_U(v)| \leq |u - v|
    \end{align*}
\end{itemize}

%Now, let $U := \{\sigma(\varphi(v)) \mid v \in S\}$.
Then, for each edge $(u, v) \in E(S)$ we get the following inequalities:
\begin{align*}
|\sigma_S(u) - \sigma_S(v)| = |ord_U(\sigma(\varphi(u))) - ord_U(\sigma(\varphi(v)))|\\ \leq |\sigma(\varphi(u)) - \sigma(\varphi(v))| \leq bandwidth(G)
\end{align*}
\end{proof}

We know that all the graphs that appear during the requests processing (revealing the edges) are subgraphs of $Ladder_n$. Thus, by Lemma~\ref{lem:embedding enumeration} we conclude that their $bandwidth \leq 2$. And we can use the embedding function from this Lemma to enumerate each subgraph $S$ of $Ladder_n$ with $\sigma_S$.

\begin{remark}
We do not need to worry about the top and bottom bounds of $Ladder_n$ when performing an embedding. In fact, we can perform an embedding of $S$ into the $Ladder_\infty$ and, since $S$ is connected and the embedding preserves connectivity, the whole image of $S$ will be within some $Ladder_m$ (for $m \geq n$) which is enough to obtain a requested $bandwidth \leq 2$. 
\end{remark}

\subsection{Connectivity component structure}
\label{ap:conncomp}

Requests come with time possibly revealing new edges of a request graph and forming connectivity components which are subgraphs of the request graph.

One connectivity component can be decomposed into cycles and trees. Let us now provide some statements about tree and cycle embedding.

\vspace{0.2cm}

\noindent{\large\textbf{Tree embedding}}

{
\def\thedefinition{\ref{def:trunk}}
\begin{definition}
Consider some correct embedding $\varphi$ of a tree $T$ into $Ladder_n$. Let $r = \arg\max\limits_{v \in V(T)} level\langle \varphi(v)\rangle$ be the ``rightmost'' node of the embedding and\\
$l = \arg\min\limits_{v \in V(T)} level\langle \varphi(v)\rangle $ be the ``leftmost'' node of the embedding. The trunk of $T$ is a path in $T$ connecting $l$ and $r$. The trunk of a tree $T$ for the embedding $\varphi$ is denoted with $trunk_\varphi(T)$.
\end{definition}
\addtocounter{definition}{-1}
}

{
\def\thedefinition{\ref{def:occupied}}
\begin{definition}
Let $T$ be a tree and $\varphi$ be its correct embedding into $Ladder_n$.
The level $i$ of $Ladder_n$ is called \emph{occupied} if there is a vertex $v \in V(T):\ \varphi(v) \in level_{Ladder_n}(i)$.
\end{definition}
\addtocounter{definition}{-1}
}

{
\def\thestatement{\ref{trunk each level}}
\begin{statement}
For every occupied level $i$ there is $v \in trunk_\varphi(T)$ such that $v \in level(i)$.  
\end{statement}
\addtocounter{statement}{-1}
}
\begin{proof}
By the definition of the trunk, an image goes from the minimal occupied level to the maximal. It cannot skip a level since the trunk is connected and the correct embedding preserves connectivity.
\end{proof}

The trunk of a tree in an embedding is an useful concept to define since the following holds for it.

{
\def\thelemma{\ref{lem:left line-graphs}}
\begin{lemma}
Let $T$ be a tree correctly embedded into $Ladder_n$ by some embedding $\varphi$. Then, all the connected components in $T \setminus trunk_\varphi(T)$ are line-graphs. 
\end{lemma}
\addtocounter{lemma}{-1}
}
\begin{proof}
% \label{lem:left line-graphs}
%[Proof of Lemma \ref{lem:left line-graphs}]
%\label{proof:lem:left line-graphs}
Suppose that it is not true and then there should exist a subgraph $S$ of $T$ such that $V(S) \cap V(trunk_\varphi(T)) = \emptyset$ and $S$ contains a node of degree three. Since there is a node of degree three in $S$ we can state that there are two nodes of $S$, say $u$ and $v$ with the same level ($level\langle\varphi(v)\rangle = level\langle\varphi(u)\rangle$). But the image of the tree trunk passes through all occupied levels of the grid by Statement \ref{trunk each level}. Hence, either $u$ or $v$ $\in trunk_\varphi(T)$ which contradicts the assumption.
\end{proof}

The bad thing about the trunk is that it depends on the embedding. And there can be several correct embeddings of the same tree giving different trunks. So, we introduce the concept of a \textit{trunk core} which alleviates this issue. But at first, we prove some technical statements.

\begin{statement}\label{pass through three node}
For the tree $T$, disregarding the correct embedding $\varphi$, the $trunk_\varphi(T)$ must pass through a node of degree three if it has no neighbours of degree three. If there are two adjacent nodes with degree three, the trunk must pass through at least one of them. 
\end{statement}
\begin{proof}
First, consider the case of a node with no neighbours of degree three. Let's call it $a$. To prove by contradiction we assume that the trunk does not pass through $a$. Let's call $a$-s neighbours $b, c$ and $d$. W.l.o.g assume that
\begin{align}
    \varphi(a) = level(i)[1]\\
    \varphi(b) = level(i - 1)[1]\\
    \varphi(c) = level(i)[2]\\
    \varphi(d) = level(i + 1)[1]
\end{align}

\va{picture}

Since the trunk does not pass through $a$ and by Statement~\ref{trunk each level} it passes through the level $i$ it should pass through $c$. If $c$ has degree one, the trunk contains one node from level $i$, and does not contain any node from $i + 1$, thus, this trunk cannot contain the topmost node.  If $c$ has degree two, we say that its second neighbour is mapped to $level(i - 1)[2]$. The case when it is mapped to $level(i + 1)[2]$ is symmetric. But then trunk does not pass through the $i + 1$ level which contradicts the Statement~\ref{trunk each level}.

Now, coming to the case with two adjacent nodes of degree three, we have two adjacent nodes $a$ and $b$ of degree three. And let $c, d$ be the rest neighbours of $a$ and $e, f$ be the rest neighbours of $b$. If $a$ and $b$ are embedded to the same level, then by Statement~\ref{trunk each level} the trunk passes through at least one of them. Suppose now that $a$ and $b$ are on different levels, say
\begin{align}
    \varphi(a) = level(i)[1]\\
    \varphi(b) = level(i + 1)[1]\\
    \varphi(c) = level(i)[2]\\
    \varphi(d) = level(i - 1)[2]\\
    \varphi(e) = level(i + 1)[2]\\
    \varphi(f) = level(i + 2)[1]
\end{align}

Since the edge $(a, b)$ is a bridge between two connected components of a tree and the trunk contains nodes form both components the trunk should pass through the edge $(a, b)$, so it passes through both $a$ and $b$.

\va{pictures}
\end{proof}

{
\def\thelemma{\ref{trunk certain nodes}}
\begin{lemma}
For the tree $T$ for each node $v$ of degree three (except for maximum two of them) we can verify in polynomial time if for any correct embedding $\varphi$ $trunk_\varphi(T)$ passes through $v$ or not. \end{lemma}
\addtocounter{lemma}{-1}
}
\begin{proof}%[trunk certain nodes]
%[Proof of Lemma \ref{trunk certain nodes}]
%\label{proof:trunk certain nodes}
\va{Possibly, rewrite the proof.}
We call a pair of adjacent nodes of degree three \emph{``paired''} nodes. We call a node of degree three with no neighbours of degree three \emph{``single''}.

If the tree contains not more than two nodes of degree three, the statement is trivial. So, we suppose that there exist at least three nodes of degree three.

The trunk passes through the single nodes by Statement~\ref{pass through three node}. Thus we are interested in paired nodes. Consider such pair. Let's call its nodes $a$ and $b$. By the Statement~\ref{pass through three node} we know that either $a$ or $b$ is in the trunk. 

By the assumption there exist either another single node or other paired nodes. If there is a single node, let's call it $c$, we know that it is in the trunk. $c$ is reachable from $a$ and $b$ and since we have tree either $a$ is on path from $b$ to $c$ or $b$ is on path from $a$ to $c$. W.l.o.g. assume $b$ is on a path from $a$ to $c$. But this implies that $b$ is in the trunk, because if not, $a$ is, and, thus, there are two paths from $a$ to $c$~--- the trunk and the one containing $b$. Thus, we have a cycle, which is impossible since we have a tree.

If there are no single nodes, there are paired nodes. We denote them with $u$ and $v$. $u$ and $v$ are reachable from $a$, thus, w.l.o.g. we can assume that $u$ is on the path from $a$ to $v$. If now $b$ is on the path from $a$ to $u$, we have the following: $a \rightarrow b \rightsquigarrow u \rightarrow v$. By Statement~\ref{pass through three node} we know that the trunk must pass through either $u$ or $v$. Denote the one the trunk passes through with $c$. We can reduce this case to the previous one, if we take any $c = u$ or $c = v$. Applying the same reasoning we deduce that $b$ must be in the trunk.

\emph{Support nodes} are the nodes of two types: either it is a single node, or it is a node that is located on the path between two other nodes with degree three. This Lemma shows that the support nodes appear in the trunk of every correct embedding.

We make a path $P$ through support nodes. For any inner node of this path which is paired there is no chance for its pair to be in a trunk if it is not in $P$ already, because the trunk is a path. So, the only uncertainty remains about at most one node in the pairs of end nodes.
\end{proof}

\begin{definition}
Path $P$ constructed in Lemma~\ref{trunk certain nodes} is called \emph{trunk core}. We denote this path for a tree $T$ as $trunkCore(T)$. Note that it can be embedded into $Ladder_n$.
\end{definition}

\begin{definition}
The embedding $\varphi$ of a line-graph $l$ on the grid is called monotone if the nodes $\varphi(l[i])$ and $\varphi(l[j])$ are on the same level of the grid only when they are adjacent on $T$. 
\end{definition}

\begin{lemma}\label{monotone end-nodes}
If a line graph is embedded preserving edges into $Ladder_n$ with no self-intersections non-monotonically then one of the end-points shares a level with a node of a path it is not adjacent with.
\end{lemma}
\begin{proof}
Denote a line graph $l$. Let's say $i$ is the smallest index such that $l[i]$ shares level with some other non-adjacent node $l[j]$, $|i - j| > 1$. W.l.o.g. let's assume that $l[i]$ is embedded to $level(k)[1]$. Since $i$ was chosen the smallest $j > i$. Let us assume that $l[i - 1]$ is embedded to $level(k - 1)[1]$. Then, since $l[j]$ is embedded to $level(k)[2]$, $l[i + 1]$ is embedded into $level(k + 1)[1]$. We also state that $l[j - 1]$ is embedded to $level(k + 1)[2]$, since if it is embedded into $level(k - 1)[2]$, the path should go from $l[i + 1]$ to $l[j - 1]$ (note that $i + 1 < j - 1$) without passing through level $k$ which is impossible. So we have the following embeddings:
\begin{align}
    l[i] \rightarrow level(k)[1]\\
    l[i - 1] \rightarrow level(k - 1)[1]\\
    l[j] \rightarrow level(k)[2]\\
    l[j + 1] \rightarrow level(k - 1)[2]
\end{align}

It is easy to see that $l[j + 2]$ has no other options but to be embedded to $level(k - 2)[2]$. But then $l[i - 3]$ should be embedded to $level(k - 3)[1]$ and so on $l[j + t] \rightarrow level(k - t)[2]$ and $l[i - t] \rightarrow level(k - t)[2]$ in general. We now take $t = \min(i - 1, length(p) - j)$ so either $l[i - t]$ or $l[j + t]$ is an end nodes and they both exist. They share level, so the lemma is proved. 
\end{proof}

\begin{lemma}\label{trunk core monotone}
The trunk core of a tree $T$ is always embedded in the monotone manner. 
\end{lemma}
\begin{proof}
Trunk core connects nodes of degree three which cannot be embedded with any other nodes of the trunk to the same level since then either a cycle appears or we obtain a conflict. Thus, by the Lemma~\ref{monotone end-nodes} trunk core must be embedded monotonically. 
\end{proof}

From now on we assume that every mentioned tree can be embedded into the $Ladder_n$. 

{
\def\thedef{\ref{def:simple-graphs}}
\begin{definition}
Let $T$ be a tree. All the connectivity components in $T\,\setminus\,trunkCore(T)$ are called \emph{simple-graphs} of tree $T$.
\end{definition}
\addtocounter{lemma}{-1}
}

{
\def\thelemma{\ref{simple-graphs line-graphs}}
\begin{lemma}
Simple-graphs of a tree $T$ are line-graphs.
\end{lemma}
\addtocounter{lemma}{-1}
}
\begin{proof}%[simple-graphs line-graphs]
%[Proof of Lemma \ref{simple-graphs line-graphs}]
%\label{proof:simple-graphs line-graphs}
Note that all the nodes of degree three in $T$ are either in the trunk core or they are adjacent to the trunk core. Hence after removing the nodes of the trunk core no nodes of degree three are left and, thus, all the graphs left are line-graphs.
\end{proof}

{
\def\thedefinition{\ref{def:hands}}
\begin{definition}
The edge between a simple-graph and the trunk core is called a \emph{leg}.

The end of a leg in the simple-graph is called a \emph{head} of the simple-graph.

The end of a leg in the trunk core is called a \emph{foot} of the simple-graph.

If you remove the head of a simple-graph and it falls apart into two connected components, such simple-graph is called \emph{two-handed} and those parts are called its \emph{hands}. Otherwise, the graph is called \emph{one-handed}, and the sole remaining component is called a \emph{hand}. If there are no nodes in the simple-graph but just a head it is called zero-handed.

\begin{center}
\includegraphics[scale=0.3]{images/simple-graph.png}
\end{center}
\end{definition}
\addtocounter{definition}{-1}
}

{
\def\thedefinition{\ref{def:exit}}
\begin{definition}
A simple-graph connected to the end nodes of the trunk core is called \emph{exit-graph}.
\end{definition}
\addtocounter{definition}{-1}
}

{
\def\thedefinition{\ref{def:inner}}
\begin{definition}
A simple-graphs connected to the inner nodes of the trunk core is called \emph{inner-graph}.
\end{definition}
\addtocounter{definition}{-1}
}

Please note that the next definition is about a much larger ladder $Ladder_N$ rather than $Ladder_n$. $N$ should be approximately equal to $2 \cdot n$.

{
\def\thedefinition{\ref{def:quasi-correct}}
\begin{definition}
An embedding $\varphi: V(G) \rightarrow V(Ladder_N)$ of a graph $G$ into $Ladder_N$ is called \emph{quasi-correct} if:
\begin{itemize}
    \item $(u, v) \in E(G) \Rightarrow (\varphi(u), \varphi(v)) \in E(Ladder_N)$, i.e., images of adjacent vertices in $G$ are adjacent in the grid.
    \item There are no more than \textbf{three} nodes mapped into each level of $Ladder_N$, i.e., the two grid nodes on each level are the images of no more than three nodes.
\end{itemize}
\end{definition}
\addtocounter{definition}{-1}
}

We might think of a quasi-correct embedding as an embedding into levels of the grid with no more than three nodes embedded to the same level. We then can compose this embedding with an embedding of a grid into the line which is the enumeration level by level. More formally if a node $u$ is embedded to the level $i$ and a node $v$ is embedded to the level $j$ and $i < j$ then the resulting number of $u$ on the line is smaller then the number of $v$, but if two nodes are embedded to the same level, we give no guarantee. 

{
\def\thelemma{\ref{quasi const cost}}
\begin{lemma}
For any graph mapped into the ladder graph by the quasi-correct embedding as described above can be mapped onto the line level by level with the property that any pair of adjacent nodes are embedded at the distance of at most five.
\end{lemma}
\addtocounter{lemma}{-1}
}
\begin{proof} %[quasi const cost]
%[Proof of Lemma \ref{quasi const cost}]
%\label{proof:quasi const cost}
Since two adjacent vertices are embedded to the levels with a number difference of at most 1, we can state that there are no more than 4 nodes between them in the line, since there are no more than 3 nodes per level.  
\end{proof}

\subsubsection{Tree embedding strategy}
\label{tree embedding strategy}

\noindent We start with the discussion on how to embed a tree with $|V(trunkCore(T))| \leq 1$.
Such tree can have $\leq 3$ nodes of degree three, since, otherwise, there are at least four nodes of degree three and the trunk core has at least two nodes:
\begin{itemize}
    \item $\geq 2$ single nodes. In this case, they are both in the trunk core.
    \item at least one single node and at least one paired nodes. In this case, one node from a pair and a single node are in the trunk core.
    \item at least two disjoint paired nodes. In this case, for each pair we know for certain the member who is in the trunk, thus we again have at least two nodes in the trunk core.
\end{itemize}

Further, we analyse the cases depending on the number of nodes of degree three.
We need the following technical Lemma.
\begin{lemma}\label{three of three}
If there is a tree with three nodes of degree three $a$, $b$, and $c$ and there are edges $(a, b)$ and $(b, c)$, then the third neighbour of $b$ is of degree one and for any correct embedding $a$, $b$, and $c$ are embedded to the different levels. 
\end{lemma}
\begin{proof}
Consider a correct embedding $\varphi$. Say $\varphi(b) = level(i)[1]$. If now $\varphi(a) = level(i)[2]$, both $level(i + 1)[2]$ and $level(i - 1)[2]$ are occupied by neighbours of $a$, so no matter where we embed $c$, say to $level(i + 1)[1]$ there would be only one spare slot, $level(i + 2)[1]$ in this case, for two neighbours of $c$. Recall that we have a tree so $a$ and $c$ can't share more then one neighbour. \va{picture}

So the only possible embedding up to the symmetry is
\begin{align}
    \begin{cases}
        \varphi(b) = level(i)[1]\\
        \varphi(a) = level(i - 1)[1]\\
        \varphi(c) = level(i + 1)[1]\\
    \end{cases}
\end{align}

In this case $level(i \pm 1)[2]$ are occupied by neighbours of $a$ and $c$, so the third neighbour of $b$ cannot have any neighbours except for $b$ since there is no place for them.

\va{picture}
\end{proof}

Now, we consider the possible cases for the amount of nodes with degree three.

\begin{itemize}
    \item There are no nodes of degree three. In this case our tree is just a line-graph $l$ and we embed it the following way: 
    \begin{align}
        \varphi: V(l) \rightarrow Ladder_n\\
        \varphi(l[i]) = level(i)[1]
    \end{align}
    Remember that right now we allow to choose any $n$, arbitrary large. 
    \item There is one node of degree three. We can think of it as two line graphs $l_1$ and $l_2$ with additional edge $(l_1[i], l_2[1])$ for some $i$. We embed the tree in the following way:
    \begin{align}
        \varphi: V(l_1) \cup V(l_2) \rightarrow Ladder_n\\
        \varphi(l_1[j]) = level(j)[1]\\
        \varphi(l_2[j]) = level(i + j - 1)[2]
    \end{align}
    \item There are two nodes of degree three. Since $|V(trunkCore(T)| \leq 1$, we conclude that those two nodes are paired, since otherwise they would be single nodes and therefore be in the trunk core by Lemma~\ref{trunk certain nodes}. So, in this case we can present $T$ as two line-graphs $l_1$ and $l_2$ with additional edge $(l_1[i], l_2[j])$ for some $i$ and $j$. We embed $T$ in the following way:
    \begin{align}
        \varphi: V(l_1) \cup V(l_2) \rightarrow Ladder_n\\
        \varphi(l_1[k]) = level(k)[1]\\
        \varphi(l_2[k]) = level(i + k - j)[2]
    \end{align}
    \item There are three nodes of degree three. \va{Why? We proved only for four nodes...} Since $|V(trunkCore(T)| \leq 1$, we conclude that there is no single node, otherwise, it is in the trunk core and one of the other two is also in the trunk core, contradicting the assumption. 
    
    So, with our three nodes of degree three, say $a$, $b$ and $c$ we must have edges $(a, b)$ and $(b, c)$. By the Lemma~\ref{three of three} none of $a,b,c$ can be embedded into the same level, and the third neighbour of $b$ is of degree one. Denote the line-graphs connected to $a$ with $l^1_{a}$ and $l^2_{a}$ (they are line-graphs since we have only three nodes of degree three) and the line-graphs connected to $c$ with $l^1_c$ and $l^2_c$. Denote the third neighbour of $b$ with $d$. We embed as follows:
    \begin{align}
        \varphi(b) = level(2)[1]\\
        \varphi(d) = level(2)[2]\\
        \varphi(a) = level(1)[1]\\
        \varphi(c) = level(3)[1]\\
        \varphi(l^1_a[i]) = level(1 - i)[1]\\
        \varphi(l^2_a[i]) = level(2 - i)[2]\\
        \varphi(l^1_c[i]) = level(3 + i)[1]\\
        \varphi(l^2_c[i]) = level(2 + i)[2]
    \end{align}
    \item There are no other cases, since we showed that if there are four nodes with degree three then the size of the trunk core should be bigger than one.
\end{itemize}

Now, we discuss how to embed a more generic tree with $|V(trunkCore(T))| \geq 2$ into the grid. We call our embedding as $\tilde{\varphi}$.
\begin{enumerate}
    \item $\tilde{\varphi}(trunkCore(T)[i]) = level(i)[1]$
    \item Suppose $l$ is a simple-graph connected to the inner node with number $i$ of the trunk core by its $j$-th node, so the leg of $l$ is $(trunkCore[i], l[j])$. We embed $l[j]$ to the opposite of $trunkCore[i]$, i.e. $\tilde{\varphi}(l[j]) = level(i)[2]$.
    
    We also want to reserve nodes $level(|V(trunkCore(T))|)[2]$ and $level(1)[2]$ for exit-graphs, so we say we embed phantom nodes there for algorithm not to use them on Step 3.
    
    \item We now want to embed hands of simple-graphs connected to the inner trunk core nodes. Suppose we have such simple graph $l$ and we've embedded its head to $level(i)[2]$ on Step 2.
    
    If $l$ is zero-handed, it is already embedded on Step 2.
    
    If $l$ is two-handed, denote its hands with $h_1$ and $h_2$ and choose the one of embeddings from 
    \begin{align}
        \left[ \begin{array}{c}
            \begin{cases}
                \tilde{\varphi}(h_1[j]) = level(i + j)[2]\\
                \tilde{\varphi}(h_2[j]) = level(i - j)[2]\\
            \end{cases}\\
            \begin{cases}
                \tilde{\varphi}(h_1[j]) = level(i - j)[2]\\
                \tilde{\varphi}(h_2[j]) = level(i + j)[2]\\
            \end{cases} 
        \end{array}
        \right.
    \end{align}
    which does not map nodes from $V(h_1) \cup V(h_2)$ to the place nodes were mapped to on step 2.
    
    If $l$ is one-handed, denote its hand with $h$ and consider two cases:
    \begin{itemize}
        \item 
        \begin{align}
        \begin{cases}
            trunkCore(T)[i + 1]\ \parbox[t]{.6\textwidth}{is an inner node and it is a foot of another\newline one-handed or zero-handed simple-graph $l_2$\newline with hand (possibly empty) $h_2$}\\
            m(h[j]) = level(i - j)[2]\ \parbox[t]{.6\textwidth}{ maps some nodes of $h$ to the place\newline where nodes were placed on step 2}
        \end{cases}\\
        \end{align}
        
        In this case we define $\tilde{\varphi}$ for $l$ and $l_2$ at a time the following way:
        \begin{align}
            \begin{cases}
                \tilde{\varphi}(h[j]) = level(i + j)[2]\\
                \tilde{\varphi}(h_2[j]) = level(i + 1 - j)[2]
            \end{cases}
        \end{align}
        
        \item The symmetric case is when
        \begin{align}
        \begin{cases}
            trunkCore(T)[i - 1]\ \parbox[t]{.6\textwidth}{is an inner node and it is a foot of another\newline one-handed or zero-handed simple-graph $l_2$\newline with hand (possibly empty) $h_2$}\\
            m(h[j]) = level(i + j)[2]\ \parbox[t]{.6\textwidth}{ maps some nodes of $h$ to the place\newline where nodes were placed on step 2}
        \end{cases}\\
        \end{align}
        
        In this case we define $\tilde{\varphi}$ for $l$ and $l_2$ at a time the following way:
        \begin{align}
            \begin{cases}
                \tilde{\varphi}(h[j]) = level(i - j)[2]\\
                \tilde{\varphi}(h_2[j]) = level(i - 1 + j)[2]
            \end{cases}
        \end{align}
        
        \item If the previous two cases don't come true we act pretty much the similar as we did for two-handed simple-graph, namely denote hand of $l$ with $h$ and choose one of the following definitions of $\tilde{\varphi}$ which  doesn't map nodes of $h$ to the places already used on step 2:
        \begin{align}
        \left[ \begin{array}{c}
            \tilde{\varphi}(h[j]) = level(i + j)[2]\\
            \tilde{\varphi}(h[j]) = level(i - j)[2]
        \end{array}
        \right.
    \end{align}
    \end{itemize}
    
    \item The last case is to consider an exit-graph.
    
    Denote the end-node of the trunk core, to which the exit-graph is connected by $i$. $i$ is either $|V(trunkCore(T))|$ or $1$.
    
    \begin{itemize}
        \item If $i = |V(trunkCore(T))|$. There are two line-graphs connected to $i$, say $l_1$ and $l_2$. Note that they can't both be two-handed, since that means we have three nodes of degree three in a row, the middle one is the end node of the trunk core, but then the middle one by Lemma \ref{three of three} must have the third neighbour of degree one which is not the case since it is a trunk core node and trunk core nodes are all of degree $\geq 2$. 
        
        So let's assume that $l_1$ is one-handed. We embed it with 
        \begin{align}
            \tilde{\varphi}(l_1[j]) = level(i + j)[1]
        \end{align} 
        If $l_2$ is one-handed we embed it with 
        \begin{align}
            \tilde{\varphi}(l_2[j]) = level(i + j - 1)[2]
        \end{align}
        
        If $l_2$ instead has two hands $h_1$ and $h_2$ and it connects to $i$ with the node $l_2[j]$. Then we define 
        \begin{align}
            \tilde{\varphi}(l_2[j]) = level(i)[2]\\
            \tilde{\varphi}(h_1[k]) = level(i + k)[2]\\
            \tilde{\varphi}(h_2[k]) = level(i + k)[2]
        \end{align}
        
        \item If $i = 1$ we do everything symmetrically. Remember we don't care if we go out $Ladder_n$ top or bottom borders, if it happens, we can just enlarge our grid to $Ladder_m$ for some large enough $m$ to accommodate the image.
    \end{itemize}
\end{enumerate}
    
\begin{definition}
The definition of $\varphi$ on the hand(s) of the simple-graph connected to the inner trunk core node is called the orientation of that simple-graph.
\end{definition}

\begin{definition}
We say that two inner simple-graphs are neighbours if there are no other simple-graphs connected to the trunk core in between their foots.
\end{definition}

\begin{lemma}
The resulting embedding of this strategy exists and it is quasi-correct.
\end{lemma}

Before diving into prove let us discuss what does the Lemma give to us. There are three key points about the described quasi-correct embedding.

First of all, we should emphasise that such embedding can be efficiently computed.

Not only that, but it also can be recomputed easily while remaining quasi-correctness when the new vertices come, which is relevant to the online scenario.

And last but not least, recall that each two adjacent nodes are embedded at the distance at most five (see Lemma \ref{quasi const cost}) so we are not worried about serving the same request many times. 

\begin{proof}
The lemma is obvious for the trees with $|V(trunkCore(T))| \leq 1$, since we can do a correct embedding.

Denote the resulting embedding with $\tilde{\varphi}$. We know that a correct embedding exists, denote it $\varphi$.

\begin{itemize}
    \item The described embedding of the trunk core meets no constraints, so it always exists.

    \item Let $top := |V(trunkCore(T))|$.
    
    The described embedding of the exit-graphs does not have any constraints, so it exists. Let us now focus on the exit-graphs connected to $trunkCore(T)[top]$. For each node $u$ of those exit-graphs it is true that $\tilde{\varphi}(u)$ is on the $level(top)$ or higher. There are no more than three nodes of exit-graphs per level. Simple-graphs connected to the inner trunk core nodes are not allowed to pass through $level(top)$, so since they are connected, no nodes from simple graphs are embedded into levels $\geq top$. There are no nodes of the trunk core higher than $top$ and on $level(top)$ there are only one node from our exit-graphs. The exit graphs connected to the $trunkCore(T)[1]$ are all embedded to the levels $\leq 1$, so they can't interfere with the exit-graphs connected to the $trunkCore(T)[top]$. Thus we conclude that nodes of exit-graphs connected to the $trunkCore(T)[top]$ do not violate quasi-correctness since there are no more than three nodes on their levels. The same for the exit-graphs connected to the $trunkCore(T)[1]$.

    \item Now to the two-handed inner simple-graphs. The leg of each such simple-graph for any correct embedding must be embedded horizontally, i.e. $\varphi(foot)$ and $\varphi(head)$ must be on the same level. This is since we know that the trunk core image is monotone by Lemma \ref{trunk core monotone} and it can't be if $\varphi(foot)$ and $\varphi(head)$ are on the different levels:

    Say $\varphi(foot) = level(i)[1]$ and $\varphi(head) = level(i + 1)[1]$. Then $level(i + 1)[2]$ and $level(i + 2)[1]$ are occupied by $head$ neighbours since it is of degree three. The $foot$ is also of degree three because it is an inner trunk core node with a simple-graph connected to it. Thus $level(i - 1)[1]$ and $level(i)[2]$ are occupied with its neighbours, trunk core nodes. But the node mapped to $level(i)[2]$ can't be the end node of the trunk core since then it is of degree three and the node mapped to $level(i + 1)[2]$ is its neighbour thus we obtain a cycle $\varphi^{-1}(\{level(i)[1], level(i + 1)[1], level(i + 1)[2], level(i)[2]\})$. So there is another trunk core node after it and it is inevitably mapped to $level(i - 1)[2]$ violating monotone property of the trunk core embedding.
    
    \va{picture}

    Now denote the hands of our two-handed graph with $h_1$ and $h_2$ and let's say that $\varphi(foot) = level(i)[1],\ \varphi(head) = level(i)[2]$ and $\varphi(h_1[1]) = level(i + 1)[2]$. Then $\varphi(h_1[2])$ must be $level(i + 2)[2]$ since $level(i + 1)[1]$ is occupied by the $foot$ trunk core neighbour $next$ (remember $foot$ is inner). If now $next$ is of degree three we obtain a conflict or a cycle, since $next's$ neighbour occupies $level(i + 2)[2]$. If not, $next$ is an inner trunk core node and we continue with the $level(i + 3)[2]$ for the $h_1[3]$ and $level(i + 3)[1]$ for the next trunk core node of $next$. So we do until we ran out of $h_1$ nodes. We now say that there are no nodes of degree three in 
    \begin{align}
        \{trunkCore(T)[i + j] \mid j \in [|V(h_1)|]\}
    \end{align}
    since if there is $j$ such that $trunkCore(T)[i + j]$ is of degree three, we obtain a conflict between the third neighbour of $trunkCore(T)[i + j]$ and $h_1[j]$. That means that $\tilde{\varphi}(h_1[j]) = level(i + j)$ will not place a node to the slot already occupied on step 2 of the strategy. The same for $h_2$. We call this line of reasoning the \textit{inductive argument}.
    
    \va{picture}
    
    So we proved that for each two-handed simple-graph connected to the inner node of a trunk core one of its orientations will not face conflicts with a neighbours of a trunk core nodes of degree three. Or in other words two-handed inner graphs can't violate the existence of the described embedding.
    
    \item We've shown that the quasi-correctness can't be violated on the levels $\geq top$ and $\leq 1$. So now we need to proof that it is not violated in between. 
    
    To violate the quasi-correctness we need to obtain at least four nodes per level. Since on each level between $top$ and $1$ there is a node from the trunk core and there are no nodes from exit-graphs we conclude that there must be at least three nodes of an inner-simple graphs. And note also that they must be from the different simple-graphs since we don't embed more than one node from one inner simple-graph per level. Denote those simple-graphs with $a, b, c$. Their foots are somehow ordered in the trunk core, say $foot(b)$ is between $foot(a)$ and $foot(c)$. Since simple-graphs hands conflict at some node, we conclude that either hand of $a$ crosses the $head(b)$ or a hand of $c$ crosses the $head(b)$, otherwise $a$ and $c$ just don't share nodes. W.l.o.g. hand of $a$ crosses $head(b)$. But this is only possible when $b$ and $a$ are one-handed graphs with adjacent foots and in this case their hands are oriented contrary and they only have two conflicts: $hand(a)[1]$ is embedded to the same node as $head(b)$ and $hand(b)[1]$ is embedded to the same node as $head(a)$. So $c$ can possibly participate in that conflict only if $c$ is a one-handed graph with a foot adjacent to $foot(b)$. That is because by our strategy two-handed simple-graphs do not cross other simple-graphs heads at all and the one-handed do only if their foots are adjacent. Our goal now is to show that in such setting $c$ can be oriented the other direction to avoid conflict with $b$.
    
    Note that we have three nodes of degree three and edges\\
    $(foot(a), foot(b))$, $(foot(b), foot(c))$. This is exactly the statement of Lemma \ref{three of three}, so we conclude that we have the following structure up to symmetry:
    \begin{align}
        \varphi(foot(b)) = level(i)[1]\\
        \varphi(foot(a)) = level(i - 1)[1]\\
        \varphi(foot(c)) = level(i + 1)[1]\\
        \varphi(head(b)) = level(i)[2]\\
    \end{align}
    and we know that $b$ in fact consists of one node.
    
    We still have two possibilities for $head(c)$, namely $level(i + 1)[2]$ or $level(i + 2)[1]$.
    
    If $\varphi(head(c)) = level(i + 1)[2]$, then $\varphi(c[2]) = level(i + 2)[2]$ and by the inductive argument applied to the $hand(c)$ there are no nodes of degree three in 
    \begin{align} 
    \{trunkCore(T)[i + 1 + j] \mid j \in [|V(hand(c))|]\}
    \end{align}
    so 
    \begin{align} 
    \tilde{\varphi}(hand(c)[j]) = level(i + 1 + j)[2]
    \end{align}
    won't place nodes of $c$ to the slots already occupied on step 2 of the strategy.
    
    \va{picture}
    
    If on the other hand $\varphi(head(c)) = level(i + 2)[1]$ then\\
    $\varphi(trunkCore(i + 2))$ must be $level(i + 1)[2]$. Thus $trunkCore(i + 2)$ can't be the end of the trunk core since then it is of degree three but $level(i)[2]$ is occupied by $head(b)$. So we say that $trunkCore(i + 3)$ exists and $\varphi(trunkCore(i + 3)) = level(i + 2)[2]$ and it is also not the end node since it can't be of degree three since $level(i + 2)[1]$ is occupied by the assumption by the $head(c)$. We now apply the inductive argument obtaining that there are no nodes of degree three in 
    \begin{align} 
    \{trunkCore(T)[i + 2 + j] \mid j \in [|V(c)|]\}
    \end{align}
    We also showed that $trunkCore(T)[i + 2]$ is not of degree three, so we state that $\tilde{\varphi}(hand(c)[j]) = level(i + 1 + j)$ won't place nodes of $c$ to the slots already occupied on step 2 of the strategy.
    
    \va{picture}
    
    This completes the proof of quasi-correctness of the embedding.
    
    \item So the last thing to show is that the described embedding exists for one-handed inner graphs.
    
    Suppose we have a one-handed inner graph $l$ with hand $h$ connected to the $i$-th node of the trunk core. Suppose also that w.l.o.g. $\varphi(foot(l)) = level(i)[1]$. The only constrainted case in our strategy is when the following doesn't hold:\\ 
    \begin{align}
    \left[\begin{array}{c c}
    \begin{cases}
        trunkCore(T)[i + 1]\ \parbox[t]{.4\textwidth}{is an inner node and it is a foot of another one-handed or zero-handed simple-graph $l_2$ with hand (possibly empty) $h_2$}\\
        m(h[j]) = level(i - j)[2]\ \parbox[t]{.4\textwidth}{ maps some nodes of $h$ to the place where nodes were placed on step 2}
    \end{cases} & (1)\\
    \begin{cases}
        trunkCore(T)[i - 1]\ \parbox[t]{.4\textwidth}{is an inner node and it is a foot of another one-handed or zero-handed simple-graph $l_2$ with hand (possibly empty) $h_2$}\\
        m(h[j]) = level(i + j)[2]\ \parbox[t]{.4\textwidth}{ maps some nodes of $h$ to the place where nodes were placed on step 2}
    \end{cases} & (2)\\
    \end{array}\right.
    \end{align}
    
    For the proof by contradiction assume now that both $\tilde{\varphi}(h[j]) = level(i + j)[2]$ and $\tilde{\varphi}(h[j]) = level(i - j)[2]$ map the nodes of $h$ to the places already used on step 2. By the inductive argument that means that there exist such $j_1, j_2 \leq |V(h)|$ that $trunkCore(i + j_1)$ and $trunkCore(i - j_2)$ are of degree three. But that means that $\varphi(head(l)) \neq level(i)[2]$ since in that case by the inductive argument $\varphi(h[j])$ must be either $level(i - j)[2]$ or $level(i + j)[2]$ but in the first case we obtain a conflict with a neighbour of $trunkCore(i + j_1)$ and in the second case we obtain a conflict with $trunkCore(i - j_2)$. 
    
    So, $\varphi(head(l))$ is either $level(i + 1)[1]$ or $level(i - 1)[1]$. Let's consider the case $level(i + 1)[1]$, the second is totally symmetric. 
    
    The trunk core nodes adjacent to $foot(l)$ are $trunkCore(T)[i - 1]$ and $trunkCore(T)[i + 1]$ they are mapped by $\varphi$ to $level(i - 1)[1]$ and $level(i)[2]$. We consider the case where $\varphi(trunkCore(T)[i - 1]) = level(i - 1)[1]$ and $\varphi(trunkCore(T)[i + 1]) = level(i)[2]$ and we show that in this case (1) holds. Symmetrically  $\varphi(trunkCore(T)[i - 1]) = level(i)[2]$ and $\varphi(trunkCore(T)[i + 1]) = level(i - 1)[1]$ will lead to (2).
    
    We now prove that $trunkCore(T)[i + 1]$ can't be the end node of the trunk core.  
    
    In Lemma \ref{trunk certain nodes} we make a path through the support nodes. If\\
    $trunkCore(T)[i + 1]$ is not a support node, it can't be the end node of the trunk core since the trunk core connects to support nodes. It is neither a single node, since it a has a neighbour $trunkCore(T)[i]$ of degree three. So, since it is in the trunk core, we deduce that it is of degree three and there are nodes $a$ and $c$ of degree three s.t. there is a path $a \rightarrow trunkCore(T)[i + 1] \rightsquigarrow c$. If $a$ is different from $trunkCore(T)[i]$ then we have three consecutive nodes $trunkCore(T)[i]$, $trunkCore(T)[i + 1]$ and $a$ of degree three, so by the Lemma \ref{three of three} $\varphi(trunkCore(T)[i])$ is on the same side as $\varphi(trunkCore(T)[i + 1])$, which contradicts the assumption of $\varphi(trunkCore(T)[i]) = level(i)[1]$ and $\varphi(trunkCore(T)[i + 1]) = level(i)[2]$. So there is a node $c$ of degree three which is not adjacent to $trunkCore(T)[i + 1]$ and there is a path $trunkCore(T)[i] \rightarrow trunkCore(T)[i + 1] \rightsquigarrow c$. But then either $c$ or its pair (if it is paired) is in the trunk core meaning that the trunk core passes through $trunkCore(T)[i + 1]$ so it is inner.
    
    Thus the $trunkCore(T)[i + 2]$ exists and it has no other options but to be embedded to $level(i + 1)[2]$ since if it is embedded to $level(i - 1)[2]$ it is embedded to the same level as $trunkCore(T)[i - 1]$ and it violates the trunk core monotone property stated by Lemma \ref{trunk core monotone}. So we are now able to apply the inductive argument deducing that there are no nodes of degree three in $\{trunkCore(T)[i + 1 + j] \mid j \in [|V(l)|]\}$. Recall that by our proof by contradiction assumption we have that mapping $m:$ $m(h[j]) = level(i + j)[2]$ maps some nodes of $h$ to the place where nodes were placed on step 2 meaning that there is a node of degree three in $\{trunkCore(T)[i + j] \mid j \in [|V(h)|]\}$. But this implies that the node $trunkCore(T)[i + 1]$ is of degree three, it is inner, so we have an inner simple-graph connected to it, moreover this simple graph is a one-handed graph since otherwise we have three nodes of degree three: $trunkCore(T)[i]$, $trunkCore(T)[i + 1]$ and $head$ of that simple-graph, implying by Lemma \ref{three of three} that $trunkCore(T)[i]$ and $trunkCore(T)[i + 1]$ are mapped to the same side of the grid which as we know is not the case.
    
    So we have that $m:$ $m(h[j]) = level(i + j)[2]$ maps some nodes of $h$ to the place where nodes were placed on step 2 and that there is a one-handed inner simple-graph connected tot he $trunkCore(T)[i + 1]$ which is exactly the case (1).
    
\end{itemize}

\end{proof}

\subsubsection{Embedding with Cycles}

{
\def\thedefinition{\ref{def:maximal-cycle}}
\begin{definition}
A \emph{maximal} cycle $C$ of a graph $G$ is a cycle in $G$ that cannot be enlarged, i.e., there is no other cycle $C'$ in $G$ such that $V(C) \subsetneq V(C')$.
\end{definition}
\addtocounter{definition}{-1}
}

{
\def\thedefinition{\ref{def:whiskers}}
\begin{definition}
Consider a graph $G$ and a maximal cycle $C$ of $G$. The whisker $W$ of $C$ is a line graph inside $G$ such that:
\begin{itemize}
    \item $V(W) \neq \emptyset$, and $V(W) \cap V(C) = \emptyset$.
    \item There exists only one edge between the cycle and the whisker $(w, c)$ for $w \in V(W)$ and $c \in V(C)$. Such $c$ is called a \emph{foot} of $W$. The nodes of $W$ are enumerated starting from $w$.
%    \item There are no nodes of degree three in $V(W)$.
    \item $W$ is maximal, i.e., there is no $W'$ in $G$ such that $W'$ satisfies previous properties and $V(W) \subsetneq V(W')$.
\end{itemize}
\begin{center}
    \includegraphics[scale=0.3]{images/whiskers.png}
\end{center}
\end{definition}
\addtocounter{definition}{-1}
}

{
\def\thedefinition{\ref{def:adjacent-whiskers}}
\begin{definition}
Suppose we have a graph $G$ that can be correctly embedded into $Ladder_n$ by $\varphi$ and a cycle $C$ in $G$. Whiskers $W_1$ and $W_2$ of $C$ are called adjacent for the embedding $\varphi$ if \begin{align*}
    \forall i \in [\min(|V(W_1)|, |V(W_2)|]\ (\varphi(W_1[i]), \varphi(W_2[i])) \in E(Ladder_n)
\end{align*}
\end{definition}
\addtocounter{definition}{-1}
}

\begin{statement}\label{cycle occupies level}
For any correct embedding of a cycle $C$ into $Ladder_n$ each level of $Ladder_n$ is either occupied with two nodes of $C$ or not occupied at all.
\end{statement}
\begin{proof}
Suppose the contradictory, and there exists a correct embedding $\varphi$ of $C$ such that there is only one node of $C$, say $a$, on some level $i$, i.e., $level(i)[1]$. $a$ has two neighbours in $C$, which we call $b$ and $c$. W.l.o.g. we say that $\varphi(b) = level(i - 1)[1]$ and $\varphi(c) = level(i + 1)$. We define $next_{ab}(x)$ for the node $x \in V(C) \setminus \{a\}$ as the next node in $C$ for $x$ in the direction $ab$. It is easy to see that if $level\langle\varphi(x)\rangle > i$ then $level\langle\varphi(next_{ab}(x))\rangle > i$ since it cannot be less than $i$ and due to the connectivity of the cycle image and it cannot be equal to $i$ since then $next_{ab}(x) = a$ and then $x = c$ but $level\langle\varphi(c)\rangle = i - 1$. $level\langle \varphi(b) \rangle = i + 1 \rightarrow level\langle \varphi(next_{ab}(b)) \rangle > i \rightarrow level\langle \varphi(next_{ab}(next_{ab}(b))) \rangle > i \rightarrow \ldots \rightarrow level\langle \varphi(c) \rangle > i$ which is a contradiction.
\end{proof}

{
\def\thelemma{\ref{adjacent whiskers}}
\begin{lemma}
Suppose we have a graph $G$ that can be correctly embedded into $Ladder_n$ and there exists a maximal cycle $C$ in $G$ with at least $6$ vertices with two neighbouring whiskers $W_1$ and $W_2$ of $C$, i.e., $(\text{foot}(W_1), \text{foot}(W_2)) \in E(G)$. Then, $W_1$ and $W_2$ are adjacent in any correct embedding of $G$ into $Ladder_N$.
\end{lemma}
\addtocounter{lemma}{-1}
}
\begin{proof} %[adjacent whiskers]
%[Proof of Lemma \ref{adjacent whiskers}]
\label{proof:adjacent whiskers}
At first, we show that for every correct embedding $foot(W_1)$ and $foot(W_2)$ are embedded to the same level of the grid. Suppose not. So there exists a correct embedding $\varphi$ of $G$ s.t. $\varphi(foot(W_1)) = level(i)[1]$ and $\varphi(foot(W_2)) = level(i - 1)[1]$. By the Statement \ref{cycle occupies level} $level(i)[2]$ and $level(i - 1)[2]$ are also occupied with nodes from cycle. So we deduce that $\varphi(W_1[1]) = level(i + 1)[1]$ and $\varphi(W_2[1]) = level(i - 2)[1]$. But by Statement~\ref{cycle occupies level} it means that there are no nodes of $C$ mapped to the levels $i + 1$ and $i - 2$ and so due to connectivity of the cycle image there are no more nodes of the cycle, but then there are only four nodes in $C$.

Now we want to show that $W_1[1]$ and $W_2[1]$ are embedded to the same level of the grid for any correct embedding of $G$. Suppose not. So there exists a correct embedding $\varphi$ of $G$ s.t.
\begin{align}
    \varphi(foot(W_1)) = level(i)[1]\\
    \varphi(foot(W_2)) = level(i)[2]\\
    \varphi(W_1[1]) = level(i + 1)[1]\\
    \varphi(W_2[1]) = level(i - 1)[2]
\end{align}

But this by Statement \ref{cycle occupies level} implies that there are no nodes of $C$ mapped to levels $i + 1$ and $i - 1$ and thus there are no more nodes of $C$ at all due to the connectivity of the cycle image. Contradiction, since there are at least $6$ nodes in $C$, not $2$.

So for every correct mapping $\varphi$ of $G$ we know that up to symmetry it does the following:
\begin{align}
    \varphi(foot(W_1)) = level(i)[1]\\
    \varphi(foot(W_2)) = level(i)[2]\\
    \varphi(W_1[1]) = level(i + 1)[1]\\
    \varphi(W_2[1]) = level(i + 1)[2]
\end{align}

So there is no other option for $W_1[2]$ and $W_2[2]$ but to be embedded to $level(i + 2)[1]$ and $level(i + 2)[2]$ respectively and so until we reach the end of either $W_1$ or $W_2$. In other words for any correct embedding $\varphi$  
\begin{align}
\forall i \in [\min(|V(W_1)|, |V(W_2)|]\ (\varphi(W_1[i]), \varphi(W_2[i])) \in E(Ladder_n)
\end{align}

\end{proof}

\begin{remark}
Due to Lemma~\ref{adjacent whiskers} if the cycle is of length $\geq 6$ we can forget about an embedding while talking about adjacent whiskers.
\end{remark}

{
\def\thedefinition{\ref{def:frame}}
\begin{definition}
Assume we have a graph $G$ and a maximal cycle $C$ of length at least $6$. The frame for $C$ is a subgraph of $G$ induced by vertices of $C$ and $\{W_1[i], W_2[i] \mid i \in [\min(|V(W_1)|, |V(W_2)|)]\}$ for each pair of adjacent whiskers $W_1$ and $W_2$.
Adding all the edges $\{(W_1[i], W_2[i]) \mid i \in [\min(|V(W_1)|, |V(W_2)|)]\}$ for each pair of adjacent whiskers $W_1$ and $W_2$ makes frame \emph{completed}. 

\begin{figure}
\centering
    \includegraphics[scale=0.3]{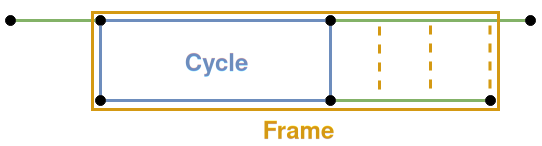}
    \caption{Cycle, its frame, and edges (dashed) to make the frame completed}
\end{figure}
\end{definition}
\addtocounter{definition}{-1}
}

\begin{lemma}\label{frame end nodes}
If we have a cycle of length at least $6$ in a graph which is a subgraph of the request graph then its end nodes of the frame are adjacent in the request graph.
\end{lemma}
\begin{proof}
This is because by Lemma \ref{adjacent whiskers} they are adjacent for every embedding and in particular for the original embedding of the cycle into the request graph.
\end{proof}

\begin{remark}
Due to Lemma~\ref{frame end nodes} we can ``extend'' each maximal cycle to the ends of its frame, so, we we do not have any adjacent whiskers, i.e., one whisker is embedded fully.
\end{remark}

\begin{lemma}\label{two nodes connected to frame}
Assume we have a graph $G$ which can be embedded into $Ladder_n$ and a maximal cycle $C$ of length at least $6$ of $G$ has no adjacent whiskers. \va{Is it after improving up to the ends.} \ap{Not exactly. It is more about the structure of "improved" cycle, not the reason why it is improved}Then, there are at most two nodes connected to $C$ (i.e. $(v, c) \in E(G)\ , such\ that \ v \in V(G) \setminus V(C) \wedge c \in V(C)$).

Moreover, these two connecting nodes are not adjacent.
\end{lemma}
\begin{proof}
Consider a correct embedding $\varphi$ of $G$ into $Ladder_n$. The cycle occupies level from $i$ to $j$, $i < j$ (it can't make a gap due to the connectivity of the image and by the Lemma \ref{cycle occupies level} it occupies the whole level). So the possible places for $v$ are $level(i - 1) \cup level(j + 1)$.

For the proof by contradiction assume that there are at least three nodes connected to $C$. Then by the pigeon hole principle there are two of them on the same level, say $v_1$ and $v_2$. There can't be an edge between $v_1$ and $v_2$ since then the cycle can be extended by adding $v_1$ and $v_2$ and thus is not maximal. But if there is no edge between $v_1$ and $v_2$ they form whiskers and those whiskers are adjacent. Contradiction.
\end{proof}

{\large \textbf{Trees and cycles}}

\begin{definition}
\va{Again, here is maximal cycle is the frame, right?} \ap{Not really. I mean the maximal in the sense of inclusion. Meaning it is not contained in any other cycle.} 
By the cycle-tree decomposition of a graph $G$ we mean a set of maximal cycles $\{C_1,\ldots C_n\}$ of $G$ and a set of trees $\{T_1,\ldots,T_m\}$ of $G$ such that 
\begin{itemize}
    \item $\bigcup\limits_{i \in [n]}V(C_i) \cup \bigcup\limits_{i \in [m]}V(T_i) = V(G)$
    \item $V(C_i) \cap V(C_j) = \emptyset\ \forall i \neq j$
    \item $V(T_i) \cap V(T_j) = \emptyset\ \forall i \neq j$
    \item $V(T_i) \cap V(C_j) = \emptyset\ \forall i\in [m],j\in[n]$
    \item $\forall i \neq j\ \forall u \in V(T_i)\ \forall v \in V(T_j)\ (u, v) \notin E(G)$
\end{itemize}
\end{definition}

\begin{lemma}\label{cycle tree trunk}
Assume we have a graph $G$ which can be embedded into $Ladder_n$. Suppose that there is a cycle $C$ and a tree $T$ from cycle-tree decomposition of $G$, such that $C$ and $T$ are connected by an edge $(c, t) \in E(G)$, where $t \in V(T)$ and $c \in V(C)$. Then, for any correct embedding $\varphi$ $t \in trunk_\varphi(T)$.
\end{lemma}
\begin{proof}
By Lemma \ref{cycle occupies level} there is another node of $C$ on the $level\langle \varphi(c) \rangle$. Let's say this level has number $i$. If $\varphi(t) \in level(i + 1)$ then we know that no nodes of $T$ can be embedded to the $level(i)$ (and thus, due to the connectivity of $T$-s image, below it) so $t$ is the bottom most node and thus it is in the trunk. Symmetrically, it is the top most node of $T$ if $\varphi(t) \in level(i - 1)$.
\end{proof}

\begin{definition}
We call a node $t$ from Lemma \ref{cycle tree trunk} an \emph{end-node} of a tree $T$, and a node $c$ a \emph{foot} of a $T$.
\end{definition}

\begin{remark}
If the tree $T$ is connected to a cycle, then $trunkCore(T)$ can be extended to an end-node of $T$.

We call the path connecting the end-node of the trunk core and an end-node of a tree an extension of the trunk core.

We call a trunk core with two of its possible extensions an extended trunk core.

The exit-graphs are now simple-graphs connected to the end-nodes of an extended trunk core.

Note that the end-nodes of the tree might not exist while the end-nodes of the trunk core are just the end-nodes of the path.
\end{remark}

We now define how to embed a tree $T$ from a cycle-tree decomposition. 

We include possible foots of a tree with their neighbours in that tree, making them the end-nodes of the trunk core. We then apply strategy~\ref{tree embedding strategy} to the obtained tree. \va{Very strange reference to the strategy.}

\begin{definition}
We say that such an embedding of a tree respects the strategy~\ref{tree embedding strategy}.
\end{definition}

\subsection{Dynamic algorithm}
Now, we talk about how we update the embedding with respect to new requests.

In our strategy of edge processing, if an already known edge is requested we do nothing since the requested nodes are already at the distance at most $12$, because by the assumption the enumeration preserves the proximity property (see the strategy plan \ref{strategy}).\va{Never told about $12$.}

But if we obtain a new edge, our enumeration may no longer maintain the proximity property \va{what property?}, so we perform a re-enumeration. \ap{Made a fix, but not sure if it is good enough yet...}

There are two possible cases for the new edge. It may be within the connectivity component or it may connect two different connectivity components. We analyse these cases separately.

We want to maintain the following five invariants:
\begin{enumerate} \label{embedding invariants}
    \item The embedding of any connectivity component is quasi-correct.
    \item For each tree in the cycle-tree decomposition the embedding of that tree matches the strategy \ref{tree embedding strategy}
    \item We do not have maximal cycles of length $4$ \va{Hmmm... Interesting...}
    \item Each maximal cycle does not have adjacent whiskers
    \item There are no conflicts with cycle nodes \va{What is conflict?} \ap{Extended the definition of a correct embedding with a definition of a conflict}
\end{enumerate}

\subsubsection{New edge within one connectivity component}

Assume we have a connectivity component $S$ with at least one cycle (call it $C_S$) and a quasi-correct embedding $\varphi'$ of $S$ preserving all the invariants from \ref{embedding invariants}. \va{Should make the Invariant "theorem"}

Assume the new edge connects nodes $u$ and $v$. Since $u$ and $v$ are already in a one connectivity component we conclude that there is now a cycle $C'$ containing $u$ and $v$. We consider a maximal cycle $C$ containing $C'$.

We call a graph $S$ with an edge $(u, v)$ as $S_+$.

If $C$ is of length $4$, for every two nodes $a$ and $b$ of $C$ $|level \langle \varphi'(a) \rangle - level \langle \varphi'(b) \rangle| \leq 3$ since the distance in $S$ between $a$ and $b$ is at most $3$ and $\varphi'$ preserves connectivity. Thus, since there are no more than $3$ nodes per level, we conclude that the difference between numbers of $a$ and $b$ is at most $12$, so, the proximity property is maintained and we do nothing. \va{And we do not add this edge to a graph?} \ap{No, we just ignore it.}

If $C$ is now of length at least $6$ (note that if a cycle can be embedded into $Ladder_n$ its length must be even) we consider its frame $F$ stating that it is in fact a cycle by Lemma~\ref{frame end nodes} and that it has at most two nodes connected to it by Lemma~\ref{two nodes connected to frame}.

\begin{lemma}\label{component above cycle}
Consider a graph $G$ embedded into $Ladder_n$ with some embedding $\varphi$ that respects invariants \ref{embedding invariants}. Consider a frame of a maximal cycle $C$ \va{I changed here cycle to frame.}\ap{Fare enough. We should come up with some elegant solution for this cycle-frame mess.} in $G$ embedded by $\varphi$ into levels from $i$ to $j$ ($i < j$). Then $\{v \mid v \in V(G),\ level \langle \varphi(v) \rangle > j\}$ form a connectivity component. 
\end{lemma}
\begin{proof}
By the invariant 4~\ref{embedding invariants} and Lemma~\ref{two nodes connected to frame} there is only one node $u$ connected to $C$ with $level \langle \varphi(u) \rangle > j$. So, if some path from $a$ to $b$ ($level \langle \varphi(a) \rangle > j$, $level \langle \varphi(b) \rangle > j$) goes through a node $v$ with $level \langle \varphi(v) \rangle \leq j$ it must pass through $u$ twice, meaning we can replace $a \rightsquigarrow u \rightsquigarrow v \rightsquigarrow u \rightsquigarrow b$ with $a \rightsquigarrow u \rightsquigarrow b$.
\end{proof}

We say that a group of nodes is on the same side from a cycle if they are in a one connectivity component when the cycle is removed.

\begin{lemma}\label{line-graph end to cycle}
If after the removal of \va{a cycle?} $F$ \va{F is frame, right?}, a connectivity component $S_i$ $i \in \{1, 2\}$ is a line-graph and it connects to $F$ via its end-node then all of its nodes belong to an exit-graph in $S$.  \va{Please, restate the lemma. Do not understand at all...} \ap{The bad thing is it is stated correctly... And I've expressed what I wanted to. I'm sorry. We should probably discuss it speaking}
\end{lemma}
\begin{proof}
Since by our assumption $S$ has a cycle, we state that each tree has a non-empty extended trunk core and, thus, each node of the component has only four options where to belong. It is either in a cycle, an extended trunk core, an inner-graph, or an exit-graph. So, we now prove that the first three do not happen here:
\begin{itemize}
    \item The node remains on a cycle when adding a new edge and the nodes from $S_i$ are not on the cycle in $S_+$, thus, they were not in $S$. \va{Don't get it. It in the connected component which is a line?} \ap{Yah, this bullet is pretty trivial after all. It couldn't have been on a cycle, since it is not on a cycle now. That's the logic.}
    \item Each node in the extended trunk core is either of degree three or has two edge-disjoint paths to nodes of degree three. Those properties can't disappear when adding a new edge and none of them hold in $S_+$ for nodes in $l$ thus did not them in $S$. So the nodes from $l$ were not in the extended trunk core.
    \item As shown in \ref{tree embedding strategy} for every correct embedding nodes of a simple-graphs always have a node embedded to the same level, namely a node from the trunk core. But $S_+$ has an embedding with all the nodes from $S_i$ being single on their level, and since every embedding for $S_+$ induces an embedding for $S$ the same holds for them in $S$ thus none of them is a node of an inner-graph in $S$.
\end{itemize}
\end{proof}

\begin{lemma}\label{line-graph inner to cycle}
If a connectivity component $l$ left when removing $F$ is a line-graph and it connects to $F$ via its inner node $u$ then all of its nodes belong to an exit-graph in $S$ except possibly $u$.
\end{lemma}
\begin{proof}
The proof is almost the same as in \ref{line-graph end to cycle} but with to adjustments:
\begin{itemize}
    \item The second bullet does not hold for $u$
    \item $l$ has an embedding where a node from $l$ is single on its level or with another node of $l$ and since nodes of $l$ are not the inner trunk core nodes that means that they are not nodes of inner simple-graphs.
\end{itemize}
\end{proof}

By Lemma \ref{component above cycle} there are at most two connectivity components left when removing $F$. Let's call them $S_1$ and $S_2$. We now describe how we embed $F$, $S_1$ and $S_2$.

So, imagine that we formed $F$ and it has node $f$ of degree three. We first discuss the case of $S_1$ being a line-graph connected with $f$ with its end-node. By Lemma \ref{line-graph end to cycle} we deduce that $S_1$ was a part of an exit-graph in $S$ and thus it was embedded strictly monotonically. In other words, (let's set an enumeration of $S_1$ such that $(S_1[1], f) \in E(S_+)$) $level\langle \varphi(S_1[1]) \rangle < level\langle \varphi(l[2]) \rangle < \ldots$. We then embed $F$ in the way that $f$ is embedded higher then any other node of $F$, say to $level(i)[1]$. And we embed $S_1[j]$ to $level(i + j)[1]$. If the levels of nodes of $S_1$ were decreasing we act the same way but embedding $f$ lower then other nodes of $F$ and embedding $S_1$ in decreasing order.

What if now $S_1$ is a line-graph connected to $F$ via its inner node. By the Lemma \ref{line-graph inner to cycle} we know that its hands (call them $h_1$ and $h_2$) were exit-graphs and thus were embedded monotonically, assume increasingly numerating from head. Assume head of $S_1$ was connected to $f \in V(F)$. We then embed $F$ in a way that $f$ is embedded higher then any other node of $F$ say to $level(i)[1]$. We then embed 
\begin{align}
    head(S_1) \rightarrow level(i + 1)[1]\\
    h_1[j] \rightarrow level(i + 1 + j)[1]\\
    h_2[j] \rightarrow level(i + j)[2]
\end{align}

We act symmetrically if the order on $h_i$ was decreasing.

We now want to show that we cannot face incompatibility namely that we have two line-graphs connected (no matter via their inner or end nodes) to $F$ both having increasing (decreasing) order on their hands in $S$. By Lemma \ref{line-graph end to cycle} those line-graphs were both whiskers in $S$, so we denote them $W_1$ and $W_2$. If they both had the same order that means they were in the same tree. To see this we first prove that there are at most two trees from the cycle-tree decomposition of the component can have exit-graphs. Consider the embedding satisfying invariants \ref{embedding invariants}. It induces an order on the maximal cycle of the component, since maximal cycles do not intersect and there for can be enumerated from bottom to top for example. If the tree is embedded between two consecutive cycles (say, $C_1$ and $C_2$) it must be connected to both of them. This is because they are connected with some path connecting nodes $c_1$ and $c_2$ of cycles. We consider such path that contains just two nodes from cycles, it is straightforward to see that such path can be obtained if we have an arbitrary one. This path (except $c_1$ and $c_2$) belongs to some tree $T$ in a cycle-tree decomposition. If now some tree different from $T$ (say, $T_2$) from a cycle-tree decomposition is embedded between cycles, since the component is connected it has two options: either to be connected to $T$ or to one of the cycles. It can't be connected to $T$ by the definition of the cycle-tree decomposition. Neither it can be connected to $C_1$ or $C_2$ since by the Lemma \ref{component above cycle} if that cycle is removed since $T$ and $T_2$ are on the one side of that cycle they are in the one connectivity component. But this implies that $T_2$ is connected to other cycle which is impossible due to the Lemma \ref{component above cycle}. 

The only way for a tree to have exit-graphs is to be embedded below the lowest cycle or above the highest. But exit-graphs in the highest tree are embedded increasingly and the whiskers in the lowest tree are embedded decreasingly. 

This means that$W_1$ and $W_2$ are in the same tree in $S$, or in other words, they are both on the same side from $C_S$. But now this cycle is contained in $F$ (since we only have $F, W_1, W_2$) and since $W_1$ and $W_2$ by Lemma~\ref{two nodes connected to frame} are now on the opposite sides of $F$ they are on the opposite sides to the $C_S$, which is impossible since then removing $C_S$ leaves $W_1$ and  $W_2$ in the different connectivity components which was not the case in $S$.

\begin{lemma}\label{tree to cycle}
Assume that $S_i$ $i \in \{1, 2\}$ connects to $F$ via a node $u$, $u \in V(T)$ where $T$ is a tree from the cycle-tree decomposition of $S_i$ and $S_i$ has a node of degree three. Then all the nodes from an extended trunk core of $T$ are embedded monotonically in $\varphi'$.  
\end{lemma}
\begin{proof}
All the nodes from a trunk core of $T$ belong to a trunk core in $S$. This is because $T$ is contained in some tree $T_S$ from a cycle-tree decomposition of $S$ and thus a trunk core of $T$ is contained in a trunk core of $T_S$ since all the support nodes remain support nodes when extending a tree. 

So we consider two extensions of $T$-s trunk core $ext_1$ and $ext_2$. Let's say $ext_1$ is the one which extends to $F$. Note that $ext_2$ was an extension in $S$ and thus we already have that $trunkCore(T) \cup ext_2$ is embedded monotonically by $\varphi'$. 

The $ext_1$ was either a part of an extended trunk core of $T_S$ or an exit-graph in $T_S$. This is because it is obviously couldn't have been a part of a cycle, since nodes on cycle remain on cycle when the new edge is added. Neither could it have been a part of the inner simple-graph since then the end-node of the $trunkCore(T)$ to which $ext_1$ is connected was an inner node of the trunk core meaning that there was two edge-disjoint paths from it to two support nodes. If those support nodes are in $T$ now the end-node of $trunkCore(T)$ is an inner trunk core node of $T$ which is nonsense. Otherwise it is not in $S_i$ meaning the path to it goes through $ext_1$ since it is the only path connecting $T$ and $F$. But then $ext_1$ belongs to a trunk core in $S$. So $ext_1$ was either a part of an exit-graph or a part of the $trunkCore(T_S)$.  In both cases it is embedded monotonically with $trunkCore(T) \cup ext_2$. 
\end{proof}

So assume we have $S_i$ $i \in \{1, 2\}$ which connects to $f \in V(F)$ via a node $u$, $u \in V(T)$ where $T$ is a tree from the cycle-tree decomposition of $S_i$ and $S_i$ has a node of degree three. Lemma \ref{tree to cycle} tells us that $T$-s extended trunk core was embedded monotonically, say, increasingly starting from $u$. In this case we embed $F$ the way that $f$ is embedded higher then every other node from $F$, say, to $level(i)[1]$, and we embed $j$-th node of an extended trunk core of $T$ to the $level(i + j)[1]$. All the other nodes from $S_i$ we embed the same as they were embedded by $\varphi'$ relatively to the nodes of the extended trunk core of $T$. 

We now analyze if we obtained a conflict of nodes from $S_i$ and $F$ and if $T$ is embedded respecting strategy \ref{tree embedding strategy}. 

Consider inner simple-graphs of $T$ in order they are connected to the extended trunk core of $T$ going through the extended trunk core from $u$. We state that all of them except possibly the first two were inner simple-graphs in $S$. This is because for a foot of such simple-graph there is a support node before it (one of the foots of first two simple-graphs) and a support node or an end-node after it (because it is inner in $S_+$). So those simple-graphs can not conflict with nodes of $F$ since they did not pass through the levels of foots of first two simple-graphs. They also respect the strategy \ref{tree embedding strategy}.

So we only need to orient first two simple-graphs in the way they don't conflict with nodes of $F$ and they respect the embedding strategy. This can be done, since the strategy can be applied to $T \cup \{f \text{ and its neighbours}\}$.

Our last goal in analyzing an embedding of a component with a node of degree three which starts with a tree is to show that we can't face incompatibility. If say, $S_1$ is starting with a tree $T$ with increasing order on its extended trunk core in $S$ then the order on $S_2$ (if one exists) is decreasing (meaning that $S_2$ has a decreasing order in $S$ on an extended trunk core of a tree it starts with or a decreasing order in $S$ on its hands if $S_2$ is just a line-graph). Let's say $S_1$ and $S_2$ are connected to $F$ with nodes $c_1 \in V(S_1)$ and $c_2 \in V(S_2)$ respectively.

First assume we have $S_1$ starting with a tree $T_1$ and $S_2$ starting with a tree $T_2$. Both $S_1$ and $S_2$ have a node of degree three. For $T_i$ take a node $u_i$ which is a foot of $T_i$ other then the one in $F$ if such exists or the end of the extended trunk core other then the one connected to $F$. Since $S_i$ has a node of degree three, one of those must exist. A shortest path connecting $u_1$ and $u_2$ contain both extended trunk cores of $T_1$ and $T_2$. Since no path in $\varphi'(S)$ is self-crossing the path $u_1 - u_2$ must be monotone by Lemma \ref{monotone end-nodes} since no nodes of the path can be embedded to the same level with $u_1$ and $u_2$ (the neighbour of $u_i$ is either a cycle node or an exit-graph node). Thus paths $u_1-c_1$ and $c_2-u_2$ have the same order in $\varphi'(S)$ so the extended trunk cores of $T_1$ and $T_2$ which are contained in $c_1-u_1$ and $c_2-u_2$ have opposite order.

Now consider the case when $S_1$ has a node of degree three and starts with a tree $T$ and $S_2$ is just a line-graph. Let's say that $S_1$ connects to $f_1 \in V(F)$ and $S_2$ connects to $f_2 \in V(F)$. For $T$ take a node $u_1$ which is a foot of $T$ other then the one in $F$ if such exists or the end of the extended trunk core other then the one connected to $F$. Since $S_1$ has a node of degree three, one of those must exist. And let's say we have an increasing order on the hands of $S_2$. Consider the top-most node of a hand of $S_2$ in $\varphi'(S)$, let's call it $u_2$. Note that $u_2$ is the top-most node among all $S$ and can share level only with nodes from exit-graphs.  Now consider the shortest path $u_1-u_2$. It contains an extended trunk core of $T$. Our goal now is to proof that this path is monotone, that would imply as in the previous case that $T$-s extended trunk core and $S_2$-s hands have different order. To see that it is monotone recall that $\varphi'$ produces no self-crossing paths and thus, if the path is not monotone, by Lemma \ref{monotone end-nodes} we either have $u_1$ sharing level with some other node from path which is impossible since it is either a cycle node or a trunk-core end-node. Or $u_2$ shares level with some other node from path which is also impossible since $u_2$ is the top-most node with only nodes from other exit-graphs hands possibly being on its level.\\

So we discussed what to do if $S_i$ connects to the $F$ with a node from tree from its cycle-tree decomposition. The last case is when it connects to $F$ with a node of a cycle from its cycle-tree decomposition.

So assume $S_1$ connects to a node $f_1 \in V(F)$ with a cycle $C$-s node, say $u_1$. And let's assume was the top-most node in $C$ the case when it is the bottom-most is totally symmetric. We embed $F$ the way $f_1$ becomes a bottom-most node we embed $C$ the way $u_1$ is the top most node and it is under $f_1$. We embed the rest of $S_1$ relatively to $C$ as it was embedded by $\varphi'$.

Nothing changed in $S_1$, so the invariants maintain for it. Therefor our only goal is to show that we don't face incompatibility with $S_2$.

If $S_2$ starts with a cycle $C_2$ which connects to $F$ via a node $u_2$ then we conclude that $u_2$ was the bottom-most node of $C_2$ since there is a path in which connects $u_1$ and $u_2$ without passing through other nodes of $C_1$ and $C_2$.

If $S_2$ starts with a tree we act similarly as we did in previous cases take such node $u_2$ that the path $u_1 - u_2$ contains tree-s trunk core/extended trunk core/hand of an exit graph and we chow that by Lemma \ref{monotone end-nodes} path $u_1-u_2$ should be monotone, since no nodes of it can be embedded to the same level as $u_2$ for the same reasons as before and neither can they be embedded to the same level with $u_1$ since there is a cycle node embedded to the same level as $u_1$. So if, say, $u_1$ was the top-most node of $C$ in $\varphi'(S)$ then the path is increasing and thus the trunk core/extended trunk core/hand of an exit graph is also increasing in $\varphi'(S)$ which is consistent to the fact that it connects to the top-most node of $F$.\\

All the actions described above assume that there was a cycle in $S$ already. If there wasn't we act as described below.

The new edge $(u, v)$ is in one connectivity component so there is a maximal cycle containing $u$ and $v$. Let's take a frame $F$ of that cycle. There are possibly two connectivity components left when removing $F$ from $S_+$, denote them $S_1$ and $S_2$. Let's say that $S_1$ connects to $f_1 \in V(F)$ and $S_2$ connects to $f_2 \in V(F)$. Moreover we know that $S_1$ and $S_2$ are trees. 

We embed $F$ the way that $f_1$ is the top-most node and $f_2$ is the bottom-most node. We then embed $S_1$ and $S_2$ the way they match strategy \ref{tree embedding strategy} orienting the trunk not to conflict with cycle nodes.\\
\va{This part is unreadable without the figures.} 

\subsection{Cost of the algorithm}
\label{app:cost}

We now want to analyze the cost of actions performed when serving a new edge within one component. Note that since the resulting embedding is quasi-correct the cost of serving the request is $O(1)$.

To make an amortized analysis we introduce the concept of \textit{scenarios}. The scenario is a reason for node to move in the embedding and thus change its number. Each scenario has two main properties: the number of times it can happen to a certain node (denoted with $SC_N$) and a cost payed for that node movement in this scenario (denoted with $SC_C$). So the total cost payed for node movements in this terms is bounded with
\begin{align}
    2n \cdot \sum\limits_{SC \in \mathrm{SCENARIOS}} SC_N \cdot SC_C
\end{align}
Where the first factor $2n$ arises due to the fact that $SC_N$ and $SC_C$ are defined for one particular node, but we want the total cost.

We propose that we only need to focus on the relative order changes.
\begin{lemma}\label{relative order cost}
If we have two enumerations $h_1$ and $h_2$ of graph $G$ then the cost of obtaining $h_2$ from $h_1$ via swaps is no more than
\begin{align}
    |\{(u, v) \mid u, v \in V(G),\ h_1(u) < h_1(v) \wedge h_2(u) > h_2(v)\}|
\end{align}
\end{lemma}
\begin{proof}
We can order nodes of $G$ by $h_2$. $h_1$ can then be viewed as a permutation so the statement of the Lemma can be reformulated as "the swap distance between a permutation $p$ and an identity permutation is less or equal the number of inversions in $p$".

We proof this by induction. The induction would be among the number of elements in permutations and among the number of inversions in permutations. 

The induction base is 0 inversions for each number of elements which is trivial. We also notice that if there is just one element in the permutation then this permutation can't have inversions. 

We now assume that our permutation $p$ has $n$ elements and $k$ inversions, with $k > 0$ and $n > 1$.

If now $p[1] = 1$ then the distance between $p$ and $identity_n$ is the distance between $p'$ and $identity_{n - 1}$ where $p'[i] = p[i + 1] - 1$. And since $p'$ has the same number of inversions as $p$ then by the inductive assumption it is less or equal to $k$ which is what we desire. 

If $p[1] = i$, $i \neq 1$ then we first spend $i - 1$ swaps to bring $i$ to the position 1 reducing the number of inversions by $i - 1$. And then apply the same idea with $p'$ obtaining that the distance from $p'$ to an $identity_{n - 1}$ is $k - (i - 1)$ and thus we provided the series of swaps to obtain an $identity_n$ from $p$ with $\leq k$.  
\end{proof}

\begin{lemma}\label{one component nodes movement}
Suppose $S$ is a connectivity component with a cycle embedded into $Ladder_n$ via $\varphi'$ which respects invariants \ref{embedding invariants}. Suppose we have a new edge in the connectivity component $S$. We denote $S$ with a new edge  by $S_+$. Let's call the frame of a maximal cycle containing the ends of the new edge $F$. The connectivity components left when removing $F$ from $S_+$ are $S_1$ and $S_2$. Suppose that our algorithm embeds $S_1$ above $F$.
\begin{enumerate}
    \item If there is a cycle in $F$ that was present in $S$ then all the nodes from $S_1$ were above that cycle in $\varphi'(S)$.
    \item If there is a cycle in $F$ that was present in $S$ then there is a node of cycle that is top-most for both new and old embeddings. 
    \item For each node $u \in V(S_1)$ and for each node $v \in V(S_2)$ $level\langle \varphi'(u) \rangle > level\langle \varphi'(v) \rangle$ except possibly the first two inner simple-graph nodes of $T_i$ if $S_i$ connects to $F$ with a tree $T_i$ from a cycle-tree decomposition of $S_i$.
\end{enumerate}
\end{lemma}
\begin{proof}
Assume $S_1$ connects to $F$ via edge $(f_1, c_1),\ f_1 \in V(F), c_1 \in V(S_1)$ and $S_2$ connects to $F$ via edge $(f_2, c_2),\ f_2 \in V(F), c_2 \in V(S_2)$
\begin{enumerate}
    \item We want to prove the following fact: consider component $A$ is connected to cycle $C$ with edge $(a, c)$, $a \in V(A)$ $c \in V(C)$ and it is embedded above $C$ by $\varphi_1$ and below $C$ by $\varphi_2$. Then for each path in $A$ starting from $a$ which is monotone for both $\varphi_1$ and $\varphi_2$ it has changed its orientation i.e. if it was increasing it is now decreasing and vice versa. To see this note that $a$ was a top-most node and became the bottom-most node of the path or vice versa. 
    
    For each type of $S_1$ our new edge processing strategy maintained an orientation of some monotone path in $S_1$ which can be extended remaining monotone to the node connected to the cycle. Thus by the fact above it must remain on the same side of the cycle. 
    \item Denote the cycle in statement by $C$,. Denote by $A$ the component that was above $C$ in $\varphi'(S)$ which contains $S_1$. Say it is connected to $C$ via edge $(a, c)$ $a \in V(A)$ $c \in V(C)$. Then by item 1 node $c$ is top-most in both embeddings of $S$ and $S_+$ since there is a monotonically increasing path starting from $a$ which does not pass through nodes of $C$.
    \item If $S_1$ and $S_2$ are both line-graphs from the proof of compatibility for line-graphs we know that the must be in the different trees in a cycle-tree decomposition of $S$ meaning they are separated by cycle and thus their nodes don't share levels. Moreover nodes form $S_1$ were in the top-most tree and nodes from $S_2$ were in the bottom-most tree, so indeed every node from $S_1$ was embedded higher than every node of $S_2$.
    
    For all the other cases on $S_1$ and $S_2$ in the proof of their compatibility we build a monotone path which passes through $c_2, f_2, f_1, c_1$. This path is monotonically increasing levels in $\varphi'(S)$ if $S_1$ is embedded above $F$ by our algorithm. So $level\langle \varphi'(c_2) \rangle < level\langle \varphi'(c_1) \rangle$ so our goal now is to analyze what nodes from $S_1$ can go under $level\langle \varphi'(c_1) \rangle$ (the analysis for $S_2$ is symmetric). 
    
    If $S_1$ connects to $F$ with a cycle or $S_1$ is a line-graph no nodes can go under $c_1$.
    
    To analyze the last case on $S_1$ we need the following fact: if $S_2$ exists then $V(F)$ contains a node of degree three in $S$. This is because $V(F)$ contains two nodes of degree three in $S_+$ namely the foots of opposite (non adjacent) whiskers and those nodes are not adjacent since the cycle is of length at least 6. So one of them must have been of degree three before the new edge.
    
    So assume now $S_1$ connects to $F$ with a tree and contains a node of degree three. From the analysis of compatibility for such $S_1$ we know that $c_1$ is either an extended trunk core node or an exit-graph node in $S$. If it is an exit-graph node, then $S_2$ does not exists since if it does there is a node of degree three in $V(F)$ in $S$ and thus there are two disjoint paths from $c_1$ to nodes of degree three which can't be for an exit-graph node.
    
    If now $c_1$ is in an extended trunk core node of a tree $T$. If $F$ contains a cycle that was present in $S$ then nodes from $S_1$ and $S_2$ were separated by this cycle then by item 1 all the nodes from $S_1$ were above that cycle and all the nodes from $S_2$ were bellow so the proposal holds. All the nodes from an extended trunk core of $T$ were above $c_1$ in $S$ so the only possible nodes are the nodes to go below $c_1$ are nodes from inner simple-graphs of $T$. And the inner simple-graphs starting from the third one can't cross the first two graphs heads so they can't cross $c_1$ as well.
\end{enumerate}
\end{proof}

We are now ready to analyze the cost payed by each node when a new edge in the connectivity component appears.

\begin{scenario}[Inner simple-graph reorienting]
The node falls into this scenario if it is a part of an inner simple-graph which is being reoriented.
\end{scenario}
\begin{lemma}
The ``Inner simple-graph reorienting'' scenario costs $O(n)$ for a node and happens only once for a given node.
\end{lemma}
\begin{proof}
Node can't pay more than $2n$ so $O(n)$ bound is trivial.

The inner simple-graph has only two possible orientations. We change its orientation only if the current one violates the invariants \ref{embedding invariants} thus we will never get back to it. 
\end{proof}

\begin{scenario}[First time on a cycle]
The node falls into this scenario if it was not on a cycle before the new edge and after the new edge it is.
\end{scenario}
\begin{lemma}
The ``First time on a cycle'' scenario happens at most once for each node and costs $O(n)$.
\end{lemma}
\begin{proof}
Trivial.
\end{proof}

\begin{scenario}[Cycle in the frame]
The node falls in this scenario if it is in the cycle which is a part of $F$ and is present in $S$.
\end{scenario}
\begin{lemma}
The ``Cycle in the frame'' scenario happens $O(n)$ times for each node and costs $O(1)$.
\end{lemma}
\begin{proof}
Since we have only $O(n)$ edges the $O(n)$ upper bound is trivial.

If the cycle is embedded respecting invariants \ref{embedding invariants} with a node specified to be the top most (which is the case due to the lemma \ref{one component nodes movement}) then we have only two four possibilities for a cycle to be embedded and for each to of them one can be transformed to another making each node changing the relative order only with $O(1)$ nodes from cycle which by Lemma \ref{relative order cost} gives us $O(1)$ cost per node. 

Note also that by Lemma \ref{one component nodes movement} nodes from cycle do not change relative order with nodes from $S_1$ or $S_2$ and thus the cost of their inner relative change is the only cost they need to pay. 
\end{proof}

So the nodes of $F$ fall into either \textit{First time on a cycle} scenario or to the cycle in the frame scenario. 

As for the $S_i$ nodes we need to consider 2 cases. The nodes are either a part of first two inner simple-graphs when $S_i$ starts from a tree or not.

If yes, and those node change their relative order to the nodes of $S_i$ this means that the reorientation of simple-graphs had been performed thus they fall into \textit{Inner simple-graph reorienting} scenario. If they did not change relative order to $S_i$ this means they didn't change the order with $S_j$, $j \in \{1, 2\} \setminus \{i\}$ neither with $F$ so they pay nothing. 

As for the nodes from $S_i$ that are not the part of first two inner simple-graphs by the Lemma \ref{one component nodes movement} they didn't change the relative order with $S_j$, by the embedding strategy they didn't change the relative order with the nodes from $S_i$ (// consider mirroring? //) accept possibly first two inner-simple graphs and also by the Lemma \ref{one component nodes movement} they didn't change the relative order with the cycle contained in $F$. As for the nodes of $F$ which were not on the cycle we say that they pay for all the relative order changes.

\subsubsection{New edge between two components}

We now define and analyse the behaviour of the algorithm when the edge between two connectivity components is revealed. The strategy would be to bring the larger component towards the smaller one.

\begin{scenario}[Connectivity component movement].
The node falls in this scenario if its component is a smaller of two between which the new edge is revealed.
\end{scenario}

\begin{lemma}
The ``Connectivity component movement'' scenario for a node happens $O(\log n)$ times and cost $O(n)$.
\end{lemma}
\begin{proof}
The size of the component is at least doubled.
\end{proof}

We now dive into the case analysis.

We first want to distinguish to cases: either the bigger component has a tree with a non-empty trunk core or not. 

If not that means that there are at most two nodes of degree three. If the new edge increases the number of degree three nodes from 0 to one, from 1 to 2 or from 2 to three we say that it is an individual scenario and we can allow the total reconfiguration according to the \ref{tree embedding strategy}. 
\begin{scenario}[New degree three node.]
The node falls into this scenario if the new edge increases the number of the nodes of degree three in the component from 0 to 1, from 1 to 2 or from 2 to 3. 
\end{scenario}
\begin{lemma}
The ``New degree three node'' scenario can happen $O(1)$ times and costs $O(n)$. 
\end{lemma}
\begin{proof}
Trivial.
\end{proof}

We now assume that there is a tree in the bigger component with a non-empty trunk core or a cycle with an inner edge. The smaller component can then connect to:
\begin{enumerate}
    \item The inner node of the trunk core. 
    \item The inner simple-graph node.
    \item The cycle node.
    \item The exit-graph node.
\end{enumerate}

In the following analysis the new scenario appear
\begin{scenario}[No more an exit-graph.]
The node falls into this scenario if it is no longer a part of an exit-graph.
\end{scenario}
\begin{lemma}
The ``No more an exit-graph'' scenario happens at most once and costs $O(n)$.
\end{lemma}
\begin{proof}
The closest node of degree three to the exit-graph node has only one path to another degree three node.
\end{proof}

Let us discuss all the possible cases.

\begin{enumerate}
    \item If the smaller connectivity component connects to the inner trunk core node then by Lemma \ref{simple-graphs line-graphs} it must be a line-graphs and there for to maintain the quasi-correctness of the embedding we only to choose its orientation and reorient its trunk core neighbours. The scenarios engaged here are the \textit{connectivity component movement} and \textit{inner simple-graph reorienting}.
    \item For the inner simple-graph node the analysis is basically the same as in the previous case.
    \item Cases:
    \begin{itemize}
        \item Only cycle
        \item Cycle and a line-graph
        \item Cycle and a tree with a degree three node
    \end{itemize}
    \item Cases:
    \begin{itemize}
        \item Smaller component does not make new nodes of degree three
        \item It does. Then we reorient the exit-graphs as they are no longer exit-graphs. This is the no more an exit graph scenario. 
    \end{itemize}
\end{enumerate}

\section{Embedding a general demand graph on a line graph}
\label{app:general}

Here we propose an online algorithm for a general demand graph $G$ assuming having an oracle algorithm for solving the Bandwidth problem.  
Given a long enough request sequence, namely $|\sigma| = \Omega(|E(G)| \cdot |V(G)|^2)$ the proposed algorithm has an $O(\lambda \cdot \bandwidth(G))$ competitive ratio compared with an optimal offline algorithm. 
But there is more to it. We point out that this algorithm is robust in a sense that its maximal cost for serving a request exceeds the maximal cost of processing the request (reconfigure + serve) of any online algorithm by at most the factor of $\lambda$.
Moreover, this algorithm pays at most $O(|E(G)| \cdot |V(G)|^2)$ for reconfiguration in total. 

{
\def\thetheorem{\ref{thm:arbitrary-graph}}
\begin{theorem}
Suppose we are given a graph $G$ and an algorithm $B$, that for any subgraph $S$ of $G$ outputs an embedding $c \in C_{S \rightarrow L_n}$ with the bandwidth less than or equal to $\lambda \cdot \bandwidth(G)$ for some $\lambda$. 
Then, for any sequence of requests $\sigma$ with demand graph $G$ there is an algorithm that serves $\sigma$ with total cost $O(|E(G)| \cdot |V(G)|^2 + \lambda \cdot \bandwidth(G) \cdot |\sigma|)$.
If the number of requests is $\Omega(|E(G)| \cdot |V(G)|^2)$ each request has $O(\lambda \cdot \bandwidth(G))$ amortized cost.
\end{theorem}
\addtocounter{theorem}{-1}
}
\begin{proof}
Assume we have processed $i$ requests so far. We get a demand graph built on edges $E_i = \{\sigma_0, \ldots, \sigma_i\}$. It induces a subgraph $S_i$ of $G$.
We want to maintain the invariant that each $S_i$ is embedded via $c_i \in C_{S_i \rightarrow L}$ such that bandwidth of $h$ is no greater than $\lambda \cdot \bandwidth(G)$.  
Suppose now the embedding $c_{i-1}$ of $S_{i-1}$ respects the invariant and we get a new request $\sigma_i$. $\sigma_i$ is an edge in $G$, say $(u, v)$. We have two possibilities: either $\sigma_i$ is already in $S_{i - 1}$ or not. 
If $(u, v) \in E(S_{i-1})$ then $S_{i-1} = S_i$ and since $c_{i - 1}$ respects the invariant we know that $|c_{i-1}(u) - c_{i-1}(v)| \leq \lambda \cdot \bandwidth(G)$ and hence we take $c_{i} = c_{i - 1}$.
If on the opposite $(u, v) \notin E(S_{i-1})$ we take $c_i = B(S_i)$, as an embedding of $S_i$ to the line, and reconfigure the network from scratch.

Now, we analyse the cost. We perform adjustments only for a new revealed edge and, thus, there would be no more than $|E(G)|$ reconfigurations. For each reconfiguration we make at most $|V(G)^2|$ migrations, meaning that the total cost of reconfigurations is at most $|E(G)| \cdot |V(G)|^2$. Since we serve the request after performing a reconfiguration and each configuration has a bandwidth of at most $\lambda \cdot \bandwidth(G)$ we state that we pay no more than $\lambda \cdot \bandwidth(G) \cdot |\sigma|$ to serve all the requests.
\end{proof}

We now explain how to construct an expensive request sequence, when given an online algorithm $ON$ to obtain a $\bandwidth(G)$ request processing cost lower bound.

{
\def\thelemma{\ref{lem:lower-bound:arbitrary-graph}}
\begin{lemma}
Given a demand graph $G$. For each online algorithm $ON$ there is a request sequence $\sigma_{ON}$ such that $ON$ serves each request from $\sigma_{ON}$ for a cost of at least $\bandwidth(G)$.
\end{lemma}
\addtocounter{theorem}{-1}
}
\begin{proof}
Consider the resulting numeration $\varphi$ of $V(G)$ done by $ON$ after serving $r \geq 0$ requests. By the definition of bandwidth, there are such $u, v \in V(G)$ that $|\varphi(u) - \varphi(v)| \geq \bandwidth(G)$. So, we let the next request $\sigma_{ON}[r + 1]$ be $(u, v)$ making $ON$ pay at least $\bandwidth(G)$ for that request. 
\end{proof}

% \begin{remark2}
% Using the previous lemma, it can be seen that $\lambda$ is a static-optimality factor for the case when $|\sigma| = \Omega(|E(G)| \cdot |V(G)|^2)$.
% \end{remark2}

\end{document}